\documentclass[english]{IEEEtran}
\usepackage[T1]{fontenc}
\usepackage[latin9]{inputenc}
\usepackage{xcolor}
\usepackage{array}
\usepackage{booktabs}
\usepackage{amsmath}
\usepackage{amsthm}
\usepackage{amssymb}
\usepackage{stmaryrd}
\usepackage{graphicx}
\usepackage{setspace}
\usepackage{wasysym}

\makeatletter

\newcommand*\LyXZeroWidthSpace{\hspace{0pt}}
\providecommand{\tabularnewline}{\\}

\theoremstyle{plain}
\newtheorem{thm}{\protect\theoremname}
\theoremstyle{plain}
\newtheorem{prop}[thm]{\protect\propositionname}
\theoremstyle{plain}
\newtheorem{cor}[thm]{\protect\corollaryname}
\theoremstyle{plain}
\newtheorem{lem}[thm]{\protect\lemmaname}

\allowdisplaybreaks
\usepackage{bibunits}
\usepackage{amsthm}
\usepackage{amssymb}
\newtheoremstyle{boldstyle1}    
  {\topsep}                    
  {\topsep}                    
  {\itshape}                   
  {}                           
  {\bfseries}                  
  {.}                          
  {.5em}                       
  {}                           

\newtheoremstyle{boldstyle2}    
  {\topsep}                    
  {\topsep}                    
  {}                           
  {}                           
  {\bfseries}                  
  {.}                          
  {.5em}                       
  {}                           

\newtheoremstyle{proofstyle}    
  {\topsep}                    
  {\topsep}                    
  {}                           
  {}                           
  {\itshape}                  
  {.}                          
  {.5em}                       
  {}                           

\theoremstyle{boldstyle2}       

\theoremstyle{boldstyle1}       

\renewcommand{\qed}{\hfill\scalebox{0.7}{$\blacksquare$}}
\theoremstyle{proofstyle}       
\newtheorem*{proof_temp}{Proof}

\theoremstyle{proofstyle}       

\makeatother

\usepackage{babel}
\providecommand{\corollaryname}{Corollary}
\providecommand{\lemmaname}{Lemma}
\providecommand{\propositionname}{Proposition}
\providecommand{\theoremname}{Theorem}

\begin{document}
\title{Multi-Object Posterior Computation\\
via Gibbs Sampling }
\author{Ba Tuong Vo, Ba-Ngu Vo \thanks{The authors are with the School of Electrical Engineering, Computing
and Mathematical Sciences, Curtin University, Australia (email: \{ba-tuong.vo,
ba-ngu.vo, \}@curtin.edu.au). Corresponding author:Ba Tuong Vo}\thanks{Acknowledgment: This work was supported by ARC Future Fellowship FT210100506
and Office of Naval Research Global Grant N629092512046.}}
\maketitle
\begin{abstract}
This work presents a tractable approach to multi-object posterior
computation under a generic measurement likelihood function. While
filtering is a popular solution, valuable historical information is
discarded. Posterior inference, which captures the full history of
the multi-object states, provides a more comprehensive solution but
is notoriously difficult and has received limited attention. Our proposed
approach uses Gibbs Sampling (GS) to generate samples from the multi-object
posterior. In particular, we establish that the conditional distributions
of the multi-object posterior are Bernoulli random finite sets with
explicit existence probabilities and attribute densities. These conditionals
are straightforward to evaluate and sample from, enabling the construction
of an efficient Gibbs sampler with standard convergence guarantees.
To demonstrate its versatility, we develop the first multi-scan multi-object
smoothing algorithm for superpositional measurements. Numerical experiments
show that the proposed method delivers robust performance in challenging
low-SNR scenarios where detection based smoothing deteriorates. Moreover,
posterior samples obtained from our approach provide statistical characterizations
of key variables and parameters, highlighting the advantages of posterior
inference. This approach enriches multi-object estimation techniques,
which historically lacked smoothing capabilities for non-standard
measurements.
\end{abstract}


\section{Introduction\protect\label{sec:Introduction}}

The goal of multi-object state estimation is to infer the underlying
trajectories of multiple objects from noisy observations \cite{BarShalomEd98,BlackmanBook99,MahlerBook14,VoVoNguyenShimo-Overview24}.
The key difference from traditional state estimation is that, instead
of a single vector representing the system state, the multi-object
(system) state is a finite set of vectors. As the multi-object state
evolves in time, the number of objects and their states vary. The
multi-object (system) trajectory consists of the trajectories of individual
objects, each of which may appear and disappear at different times.
Multi-object estimation has a wide range of applications across various
fields such as surveillance and tracking, field robotics, computer
vision,  healthcare and biomedical imaging \cite{BlackmanBook99,Cadena2016SLAM,MaggioCavallaro2011,Meijeringetal-12}.
Due to the unknown and random number of objects or trajectories, multi-object
estimation is far more challenging than traditional state estimation
\cite{VoVoNguyenShimo-Overview24}.

In the Bayesian framework, all information on the (system) trajectory
is captured in the \textit{posterior}--the joint probability density
of the sequence of (system) states up to the current time, conditioned
on the observation history \cite{Meditch73,Briers-10,DoucetTutorial09}.
The current marginal of this posterior, known as the \textit{filtering
density}, is a suboptimal solution commonly used to reduce the computational
complexity (but incurs information loss) \cite{AndersonMoore79},
\cite{Sarkkabook13}. In state estimation, it is taken for granted
that a trajectory is a sequence of states, but this is not the case
in multi-object systems, because the set of trajectories of the objects
is not the time-sequence of finite sets containing their states. Nonetheless,
this conceptual difficulty is circumvented by the labeled multi-object
representation \cite{Goodmanetal97}, \cite{VoConj13}, which ensures
that the multi-object trajectory can be uniquely reconstructed from
the sequence of multi-object states. Hence, in line with state estimation
theory, the multi-object trajectory can be estimated from a sequence
of multi-object state estimates, and the (labeled) multi-object posterior
captures all information, including statistical characterization of
variables/parameters, pertaining to the underlying multi-object trajectory
\cite{VoVoNguyenShimo-Overview24}. 

Multi-object posterior inference is both fundamental and valuable
in practice, but it is extremely challenging computationally. Most
contributions to multi-object estimation have been devoted to filtering,
while posterior inference received very little attention, see for
example the latest overview \cite{VoVoNguyenShimo-Overview24}. Filtering
offers computational efficiency, but only considers the information
on the current multi-object state. While this is acceptable for applications
with good Signal-to-Noise Ratio (SNR), in more challenging signal
settings such as Track-Before-Detect (TBD) \cite{buzzi2005tbd}, \cite{davey2007comparison},
\cite{kim2019comparison}, \cite{Rathnayake-TBDLMB-Sensors-19}, \cite{ristic2025othr},
filtering may not be adequate even for a single-object, necessitating
the integration of information from a sequence of states \cite{WeiYi2020multiframe}.
However, of the few multi-object posterior computation techniques
available, most are based on the standard multi-object model \cite{Vu2014},
\cite{VoVomultiscan18}, \cite{MVVS2022Multi}, which has limited
applicability. 

This work presents a tractable approach to multi-object posterior
computation under a \textit{generic observation model}. Our approach
is based on Gibbs sampling and serves as a versatile tool, especially
when no analytic solutions nor approximations of the posterior are
available. Under the standard multi-object model, the filtering density
and posterior take on analytic forms known as GLMBs \cite{VoConj13},
\cite{VoVomultiscan18}, and\textcolor{blue}{{} }Gibbs sampling has
been used to truncate these analytic expressions \cite{VoVoH2017},
\cite{VoVomultiscan18}, \cite{MVVS2022Multi}, \cite{SVVOM2022Linear}.
In this work we propose a Gibbs sampler to sample directly from the
multi-object posterior (that may not have an analytic form). In particular,
we derive the family of conditionals of the posterior that are easy
to sample and evaluate. Building on these results, we construct a
Gibbs sampler for the generally intractable multi-object posterior,
together with standard convergence results. Note that the proposed
approach also applies to cases with analytic solutions (e.g., the
standard observation model), but its strength lies in settings where
such solutions are unavailable or impractical.

To demonstrate the versatility of the approach, we use it to develop
the first multi-scan multi-object tracking algorithm for \textit{superpositional
measurements} \cite{MahlerSCPHD09}. In many applications, raw measurements
are preprocessed into detections to reduce memory and computational
load. While effective at moderate SNR, this preprocessing discards
information and becomes unreliable at low SNR. Existing multi-object
tracking solutions that directly exploits raw measurements have only
been developed for filtering, and often rely on heuristics or approximations,
which degrade significantly in challenging conditions \cite{NCM2013Computationally,PapiKim-15,LWLGZ2020Robust}.
We evaluate the effectiveness of the proposed multi-object tracker
on superpositional measurements\textcolor{violet}{{} }and compare with
related multi-object detection based trackers\textcolor{violet}{}.
Beyond trajectory estimation, the labeled multi-object representation
provides rich posterior analytics with intrinsic uncertainty quantification.
To illustrate this capability, we present some example analytics pertinent
to situational awareness, which can be readily computed from the posterior
samples. \textcolor{violet}{} \textcolor{violet}{} 

The remainder of this article is structured as follows. Section \ref{sec:Background}
reviews the relevant background. Section \ref{sec:MO-Posterior-Comp}
details the proposed multi\nobreakdash-object posterior Gibbs sampler.
Section \ref{sec:Superpositional} presents the proposed superpositional
measurement multi-object smoother and numerical studies. Section \ref{sec:Conclusion}
provides concluding remarks, while mathematical proofs are given in
the Appendix. Mathematical notation used throughout the paper is given
in Table 1. 

\begin{table}[t]
\begin{onehalfspace}
\caption{List of frequently used notations.\protect\label{tab:List-of-notations-1}\vspace{-1em}}

\end{onehalfspace}

\centering{}%
\begin{tabular}{|c|l|}
\hline 
{\small\,\,\,}{\small\textbf{Notation}} & {\small\quad{}\quad{}\quad{}\quad{}\quad{}\,\,}{\small\textbf{Description}}\tabularnewline
\hline 
{\small$1:k$} & {\small list of numbers $1,2,...,k$}\tabularnewline
{\small$x_{1:k}$ } & {\small lists $x_{1},x_{2},...,x_{k}$ }\tabularnewline
{\small$\mathbb{X}$} & {\small attribute space}\tabularnewline
{\small$\mathbb{L}$} & {\small label space }\tabularnewline
{\small$\bar{S}$} & {\small complement of (the set) $S$}\tabularnewline
{\small$\boldsymbol{S}^{L}$} & $\boldsymbol{S}\cap(\mathbb{X}\times L)$\tabularnewline
{\small$|X|$ } & {\small cardinality of (the finite set) $X$ }\tabularnewline
{\small$h^{X}$} & {\small multi-object exponential $\prod_{x\in X}h\left(x\right)$},\tabularnewline
{\small$\mathbf{1}_{A}$ } & {\small indicator function for the set $A$}\tabularnewline
{\small$\mathbf{1}_{A}^{B}$ } & {\small inclusion indicator for $B\subseteq A$}\tabularnewline
{\small$\delta_{A}[B]$} & {\small Kronecker delta: $1$ if $A=B$, and 0 otherwise}\tabularnewline
{\small$\Delta(\boldsymbol{X})$ } & {\small distinct label indicator $\delta_{|\boldsymbol{X}|}[|\mathcal{L}(\boldsymbol{X})|]$ }\tabularnewline
{\small$\left\lceil \frac{u}{v}\right\rfloor $} & $\frac{u}{u+v}$\tabularnewline
{\small$\left\langle f,g\right\rangle $} & {\small inner product $\int f(x)g(x)dx$}\tabularnewline
\hline 
\end{tabular}
\end{table}

\section{Background\protect\label{sec:Background}}

\subsection{Multi-Object State Representation}

In multi-object systems \cite{VoConj13}, \cite{VoVoNguyenShimo-Overview24},
an object is represented by a \textit{labeled state} $\boldsymbol{x}=(x,\ell)$,
where $x$ is its \textit{attribute} in an attribute space $\mathbb{X}$,
and $\ell$ is its \textit{label} in a (discrete) label space $\mathbb{L}$.
The attribute and label of a (labeled) state $\boldsymbol{x}\in\mathbb{X}\times\mathbb{L}$
are, respectively, given by the \textit{attribute projection} $\mathcal{A\!}:\!(x,\ell)\mapsto x,$
and \textit{label projection} $\mathcal{L}\!:\!(x,\ell)\mapsto\ell$.
In this work we use $\ell=(s,\iota)$, where $s$ is the time of birth
and $\iota$ is an index to distinguish objects born at the same time.
Let $\mathbb{B}_{s}$ denote the (discrete) space of labels born at
time $s$, and assume time starts at 0. Then, for the time interval
$\{0\textrm{:}k\}$, the label space $\mathbb{L}=\uplus_{s\in\{0:k\}}\mathbb{B}_{s}$. 

A \textit{trajectory} is a time-stamped a sequence of labeled states
on a given set of instants, with a common label $\ell$. Further,
at any time instance, a (\textit{labeled}) \textit{multi-object state}
$\boldsymbol{X}$ is a finite subset of the product space $\mathbb{X}\times\mathbb{L}$
with \textit{distinct labels}, i.e., no two elements of $\boldsymbol{X}$
share the same label. The labels of $\boldsymbol{X}$ are distinct
when the \textit{distinct label indicator}
\[
\Delta(\boldsymbol{X})\triangleq\delta_{|\boldsymbol{X}|}[|\mathcal{L}(\boldsymbol{X})|],
\]
equals 1, i.e., $\boldsymbol{X}$ has the same cardinality as $\mathcal{L}(\boldsymbol{X})$.
This representation allows multiple objects to share the same attribute,
e.g., $\{(x,\ell_{1})$, $(x,\ell_{2})\}$, which inevitably occurs
in mergings/collisions or in a finite state space.

Let $\boldsymbol{X}_{t}$ denote the multi-object state at time $t$.
Analogous to single-object systems, a sequence $\boldsymbol{X}_{j:k}=(\boldsymbol{X}_{j},...,\boldsymbol{X}_{k})$
of multi-object states completely characterizes the \textit{multi-object
trajectory}--the set of trajectories of the objects--on the interval
$\{j\textrm{:}k\}$, by grouping the states according to labels (as
illustrated in Figure \ref{fig:tracklabel} adapted from \cite{MVVS2022Multi}).
\begin{figure}[t]
\begin{centering}
\resizebox{88mm}{!}{\includegraphics[clip]{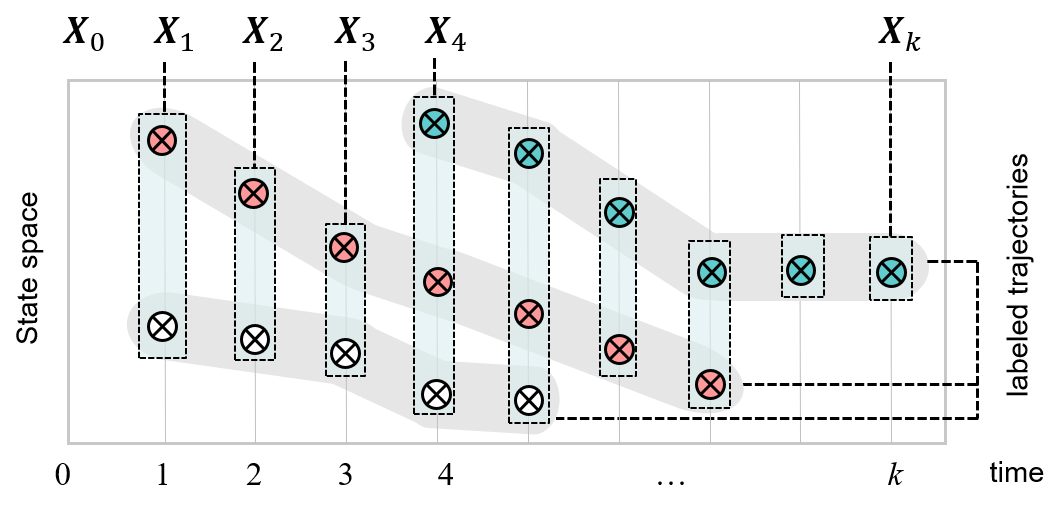}}
\par\end{centering}
\caption{\protect\label{fig:tracklabel} Multi-object trajectory and state
history. The colors of the states indicate their labels. The multi-object
state history $\boldsymbol{X}_{\negthinspace0:k}=$ ($\boldsymbol{X}{}_{\negthinspace0}$,...,$\boldsymbol{X}_{\negthinspace k}$)
can be decomposed into trajectories by grouping the states according
to labels.}
\vspace{0mm}
\end{figure}

For any $\boldsymbol{S}\subseteq\mathbb{X}\times\mathbb{L}$ and $L\subseteq\mathbb{L}$,
we denote $\boldsymbol{S}^{L}$ as the subset of $\boldsymbol{S}$
with label set $L$, i.e.,
\begin{alignat*}{1}
\boldsymbol{S}^{L}\triangleq\boldsymbol{S}\cap(\mathbb{X}\times L).
\end{alignat*}
A multi-object state $\boldsymbol{X}$ can be decomposed in terms
of its labels as $\boldsymbol{X}=\uplus_{\ell\in\mathbb{L}}\boldsymbol{X}^{\{\ell\}},$
which can be enumerated as an $|\mathbb{L}|$-tuple $\boldsymbol{X}\equiv[\boldsymbol{X}^{\{\ell_{1}\}},...,\boldsymbol{X}^{\{\ell_{|\mathbb{L}|}\}}]$.
Note that $\boldsymbol{X}^{\{\ell\}}$ is either empty or singleton,
and we denote $\boldsymbol{X}^{\bar{\{\ell\}}}$ as the subset of
$\boldsymbol{X}$ where the label set is the complement of $\{\ell\}$.

By convention,\textcolor{blue}{{} }vectors are represented by lower
case letters (e.g., $x$ and $\boldsymbol{x}$), and finite sets are
represented by upper case letters (e.g., $X$ and $\boldsymbol{X}$),
where the symbols for labeled entities and their distributions are
bolded (e.g., $\boldsymbol{x}$, $\boldsymbol{X}$, $\boldsymbol{\pi}$,
etc.) to distinguish them from unlabeled ones.

\subsection{Labeled Multi-Bernoulli}

A \textit{random finite set} (\textit{RFS}) of a space $\mathcal{X}$
is a \textit{random variable whose realizations are finite subsets
of} $\mathcal{X}$. A \textit{labeled} RFS (LRFS) with \textit{attribute
space} $\mathbb{X}$ and \textit{label space} $\mathbb{L}$, is an
RFS of the product space $\mathbb{X}\times\mathbb{L}$ such that each
realization has distinct labels \cite{VoVoNguyenShimo-Overview24}. 

A labeled multi-Bernoulli (LMB) is an LRFS parameterized by $\{(r_{\Psi}(\zeta),p_{\Psi}(\cdot;\zeta)):\zeta\in\Psi\}$
and a 1-1 mapping $\sigma:\mathbb{L}\rightarrow\Psi$ that pairs each
$\zeta\in\Psi$ with a label $\ell\in\mathbb{L}$, where $r_{\Psi}(\zeta)\in[0,1]$
and $p_{\Psi}(\cdot;\zeta)$ is a probability density on $\mathbb{X}$.
Let $\mathcal{D}(\sigma)$ denote the domain of $\sigma$, $r(\ell)\triangleq\mathbf{1}_{\mathcal{D}(\sigma)}(\ell)r_{\Psi}(\sigma(\ell))$,
and $p(\cdot,\ell)\triangleq p_{\Psi}(\cdot;\sigma(\ell))$, then
the density of an LMB can be written in terms of the exponential of
the multi-object state:
\begin{alignat}{1}
\boldsymbol{\lambda}(\boldsymbol{X};r,p) & =\Delta\left(\boldsymbol{X}\right)\left[1-r\right]^{\mathcal{D}(\sigma)-\mathcal{L}\left(\boldsymbol{X}\right)}r^{\mathcal{L}\left(\boldsymbol{X}\right)}p^{\boldsymbol{X}}.\label{eq:LMBexpX}
\end{alignat}
The parameter $r(\ell)$ is the existence probability of (the object
with) label $\ell$ and $p(\cdot,\ell)$ is the probability density
of its attribute conditional on the existence of label $\ell$. Further,
since $\sum_{L\subseteq\mathbb{L}}\left[1-r\right]^{\mathcal{D}(\sigma)-L}r^{L}=1$,
the set integral of $\boldsymbol{\lambda}$ is indeed unity \cite{VoVoNguyenShimo-Overview24}.
Note that a \textit{Bernoulli} RFS is the special case of an LMB where
the parameter set is a singleton. 

Alternatively, we can also write \cite{VoConj13}, \cite{VoVoNguyenShimo-Overview24}
\begin{alignat}{1}
\boldsymbol{\lambda}\left(\boldsymbol{X},\Psi;\lambda\right) & =\Delta(\boldsymbol{X})\mathbf{1}_{\mathcal{D}(\sigma)}^{\mathcal{L}(\boldsymbol{X})}\lambda(\boldsymbol{X};\cdot)^{\Psi},\label{eq:LMBexpParam}
\end{alignat}
where $\mathbf{1}_{\mathcal{D}(\sigma)}^{\mathcal{L}(\boldsymbol{X})}=[\mathbf{1}_{\mathcal{D}(\sigma)}]^{\mathcal{L}(\boldsymbol{X})}$,
and
\begin{alignat}{1}
\lambda(\boldsymbol{X};\zeta) & =\left\{ \!\!\!\begin{array}{ll}
r(\mathcal{\sigma}^{\texttt{-}1}(\zeta))p{}^{\boldsymbol{X}^{\{\mathcal{\sigma}^{\texttt{-}1}(\zeta)\}}}, & \!\!\!\text{if }\mathcal{\sigma}^{\texttt{-}1}(\zeta)\in\mathcal{L}\left(\boldsymbol{X}\right)\\
1-r(\mathcal{\sigma}^{\texttt{-}1}(\zeta)), & \!\!\!\text{if }\mathcal{\sigma}^{\texttt{-}1}(\zeta)\notin\mathcal{L}\left(\boldsymbol{X}\right)
\end{array}\!\!\!.\!\right.\label{eq:LMBexpParam1}
\end{alignat}
Note that $r(\mathcal{\sigma}^{\texttt{-}1}(\zeta))=r_{\Psi}(\zeta)$.
In addition, $\mathcal{\sigma}^{\texttt{-}1}(\zeta)\in\mathcal{L}\left(\boldsymbol{X}\right)$
means $\zeta$ corresponds uniquely to a label of $\boldsymbol{X}$,
i.e., there exists a unique element $\left(x,\mathcal{\sigma}^{\texttt{-}1}(\zeta)\right)\in\boldsymbol{X}$,
with label $\mathcal{\sigma}^{\texttt{-}1}(\zeta)$, and hence, $p^{\boldsymbol{X}^{\{\mathcal{\sigma}^{\texttt{-}1}(\zeta)\}}}=$
$p(x,\mathcal{\sigma}^{\texttt{-}1}(\zeta))=p_{\Psi}(x,\zeta)$. 

\subsection{Multi-Object Posterior}

In multi-object estimation, we are interested in the multi-object
trajectory or the history of the multi-object state $\boldsymbol{X}_{0:k}$.
All information on the multi-object trajectory is contained in the
multi-object posterior, defined as the density of the multi-object
state history conditioned on the observation history. The multi-object
posterior is given by \cite{VoVomultiscan18}, \cite{VoVoNguyenShimo-Overview24}:
\begin{alignat}{1}
\boldsymbol{\pi}_{z_{1:k}}(\boldsymbol{X}_{0:k}) & =\frac{\boldsymbol{\pi}_{0}(\boldsymbol{X}_{0})\prod_{t=1}^{k}\boldsymbol{g}_{t}(z_{t}|\boldsymbol{X}_{t})\boldsymbol{f}_{\!t}(\boldsymbol{X}_{t}|\boldsymbol{X}_{t\texttt{-}1})}{K(z_{1:k})},\label{eq:multiobject-posterior}
\end{alignat}
where $K(z_{1:k})$ is the normalizing constant to ensure that $\boldsymbol{\pi}_{z_{1:k}}$
integrates (via multiple set integrals) to 1, $\boldsymbol{g}_{t}(z_{t}|\boldsymbol{X}_{t})$
is the likelihood that the multi-object state $\boldsymbol{X}_{t}$
generates the observation $z_{t}$, and $\boldsymbol{f}_{\!t}(\boldsymbol{X}_{t}|\boldsymbol{X}_{t\texttt{-}1})$
is the multi-object transition density from time $t-1$ to $t$. For
simplicity hereon, we omit the dependence on the observation history
and write $\boldsymbol{\pi}_{0:k}(\boldsymbol{X}_{0:k})\triangleq\boldsymbol{\pi}_{z_{1:k}}(\boldsymbol{X}_{0:k})$.

In this work we use the following \textit{standard multi-object dynamic
model} \cite{VoConj13}, \cite{VoVoNguyenShimo-Overview24}. Note
that the space $\mathbb{L}_{t}$ of labels up to time $t$, can decomposed
into the disjoint sets $\mathbb{B}_{t}$ (the birth label space at
time $t$) and $\mathbb{L}_{t\texttt{-}1}$ (the label space up to
time $t-1$), i.e., $\mathbb{L}_{t}=\mathbb{B}_{t}\uplus\mathbb{L}_{t\texttt{-}1}$.
Consequently, the multi-object state $\boldsymbol{X}_{t}=\boldsymbol{X}_{\!t}^{\mathbb{B}_{t}}\uplus\boldsymbol{X}_{\!t}^{\mathbb{L}_{t\texttt{-}1}}$,
where $\boldsymbol{X}_{\!t}^{\mathbb{B}_{t}}$ is the set $\boldsymbol{X}_{\!t}\cap(\mathbb{X}\times\mathbb{B}_{t})$
of new born objects, and $\boldsymbol{X}_{\!t}^{\mathbb{L}_{t\texttt{-}1}}$
is the set $\boldsymbol{X}_{\!t}\cap(\mathbb{X}\times\mathbb{L}_{t\texttt{-}1})$
of surviving objects. 

Given a multi-object state $\boldsymbol{X}_{t\texttt{-}1}$ at time
$t-1$, the sets of new born and surviving objects are, respectively,
modeled by the independent LMBs $\boldsymbol{f}_{B,t}$ with parameters
$\{(P_{B,t}(\ell),p_{B,t}(\cdot,\ell)):\ell\in\mathbb{B}_{t}\}$ and
$\boldsymbol{f}_{\!S,t\!}\left(\cdot|\boldsymbol{X}_{t\texttt{-}1}\right)$
with parameters $\{(P_{S,t}(\boldsymbol{\zeta}),f_{S,t}(\cdot|\boldsymbol{\zeta})):\boldsymbol{\zeta}\in\boldsymbol{X}_{t-1}\}$.
Specifically, an object with state $(x_{t},\ell)$ $\in\mathbb{X}\times\mathbb{B}_{t}$
is either born with probability $P_{B,t}(\ell)$ and probability density
$p_{B,t}(x_{t},\ell)$, or not born with probability $Q_{B,t}(\ell)=1-P_{B,t}(\ell)$.
Further, each element $(x_{t\texttt{-}1},\ell)\in\boldsymbol{X}_{t\texttt{-}1}$
either survives with survival probability $P_{S,t}(x_{t\texttt{-}1},\ell)$
and evolves to state $(x_{t},\ell)$ at time $t$, with probability
density $f_{S,t}(x_{t}|x_{t\texttt{-}1},\ell)$, or dies with probability
$Q_{S,t}(x_{t\texttt{-}1},\ell)=1-P_{S,t}(x_{t\texttt{-}1},\ell)$.
The standard multi-object transition density is completely characterized
by the model parameters $P_{B,t},p_{B,t},P_{S,t},f_{S,t}$, and is
given by \cite{VoConj13}
\begin{equation}
\boldsymbol{f}_{\!t}\!\left(\boldsymbol{X}_{t}|\boldsymbol{X}_{t\texttt{-}1}\right)=\boldsymbol{f}_{\!B,t}\!\left(\boldsymbol{X}_{\!t}^{\mathbb{B}_{t}}\right)\boldsymbol{f}_{\!S,t}\!\left(\boldsymbol{X}_{\!t}^{\mathbb{L}_{t\texttt{-}1}}|\boldsymbol{X}_{\!t\texttt{-}1}\right).\label{eq:multiobject-trans}
\end{equation}
Note that $\mathcal{L}(\boldsymbol{X}_{\!t}^{\mathbb{B}_{t}})\subseteq\mathbb{B}_{t}$,
and $\mathcal{L}(\boldsymbol{X}_{\!t}^{\mathbb{L}_{t\texttt{-}1}})\subseteq\mathcal{L}(\boldsymbol{X}_{\!t\texttt{-}1})$.
More sophisticated multi-object dynamic models that include spawning
or splitting, and inter-object interactions can be found in \cite{MahlerBook07,MahlerBook14,Mahler2003,Bryantetal18},
\cite{Nguyenetal-CellSpawning21}, \cite{Gostar-Interacting-19}.

For simplicity, we drop the time subscript $t$, and replace the time
subscripts $t-1$ and $t+1$, by ``$-$'' and ``$+$'', respectively.
For convenience, hereon we write the LMB $\boldsymbol{f}_{B}$ as
\begin{alignat}{1}
\boldsymbol{f}_{\!B}(\boldsymbol{X}) & =\Delta(\boldsymbol{X})\mathbf{1}_{\mathbb{B}}^{\mathcal{L}(\boldsymbol{X})}Q_{B}^{\mathbb{B}-\mathcal{L}(\boldsymbol{X})}b^{\boldsymbol{X}}\label{eq:birthRFS1}\\
b(x,\ell) & =P_{B}(\ell)p_{B}(x,\ell),\label{eq:birthRFS2}
\end{alignat}
and for the LMB $\boldsymbol{f}_{\!S}\left(\cdot|\boldsymbol{X}_{\texttt{-}}\right)$
we use both forms (\ref{eq:LMBexpX}) and (\ref{eq:LMBexpParam}).\textcolor{cyan}{{}
}The first form involves a multi-object exponential of the argument
\begin{alignat}{1}
\boldsymbol{f}_{\!S}(\boldsymbol{X}|\boldsymbol{X}_{\texttt{-}}) & =\Delta(\boldsymbol{X})\mathbf{1}_{\mathcal{L}(\boldsymbol{X}_{\texttt{-}})}^{\mathcal{L}(\boldsymbol{X})}Q_{S}^{\boldsymbol{X}_{\texttt{-}}^{\mathcal{L}\left(\boldsymbol{X}_{\texttt{-}}\right)-\mathcal{L}\left(\boldsymbol{X}\right)}}\overset{\shortleftarrow}{f}(\cdot;\boldsymbol{X}_{\texttt{-}})^{\boldsymbol{X}}\label{eq:survival-RFS-X-1}\\
\overset{\shortleftarrow}{f}(x,\ell;\boldsymbol{X}_{\texttt{-}}) & =[P_{S}(\cdot)f_{S}(x|\cdot)]^{\boldsymbol{X}_{\texttt{-}}^{\{\ell\}}},\label{eq:survival-RFS-X-2}
\end{alignat}
here we use $\left[Q_{S}(\sigma(\cdot))\right]^{\mathcal{L}\left(\boldsymbol{X}_{\texttt{-}}\right)-\mathcal{L}\left(\boldsymbol{X}\right)}$
= $Q_{S}^{\boldsymbol{X}_{\texttt{-}}^{\mathcal{L}\left(\boldsymbol{X}_{\texttt{-}}\right)-\mathcal{L}\left(\boldsymbol{X}\right)}}$
because $Q_{S}(\sigma(\ell))$ = $Q_{S}^{\boldsymbol{X}_{\texttt{-}}^{\{\ell\}}}$.
The second form involves a multi-object exponential of the parameter
set--in this case, the previous multi-object state $\boldsymbol{X}_{\texttt{-}}$
:
\begin{alignat}{1}
\boldsymbol{f}_{\!S}(\boldsymbol{X}|\boldsymbol{X}_{\texttt{-}}) & =\Delta(\boldsymbol{X})\mathbf{1}_{\mathcal{L}(\boldsymbol{X}_{\texttt{-}})}^{\mathcal{L}(\boldsymbol{X})}\overset{\shortrightarrow}{f_{\texttt{-}}}(\boldsymbol{X};\cdot)^{\boldsymbol{X}_{\texttt{-}}}\label{eq:survival-RFS-Xminus-1}\\
\overset{\shortrightarrow}{f_{\texttt{-}}}(\boldsymbol{X};x_{\texttt{-}},\ell) & =\begin{cases}
\!P_{S}(x_{\texttt{-}},\ell)f_{S}(x|x_{\texttt{-}},\ell), & \!\!(x,\ell)\in\boldsymbol{X}\\
\!Q_{S}(x_{\texttt{-}},\ell), & \!\!\ell\notin\mathcal{L}(\boldsymbol{X}).
\end{cases}\label{eq:survival-RFS-Xminus-2}
\end{alignat}
Note that $\boldsymbol{f}_{B}(\boldsymbol{X})=0$ if $\boldsymbol{X}$
has label(s) outside of $\mathbb{B}$, and $\boldsymbol{f}_{\!S}(\boldsymbol{X}|\boldsymbol{X}_{\texttt{-}})=0$
if $\boldsymbol{X}$ has label(s) outside of $\mathcal{L}(\boldsymbol{X}_{\texttt{-}})$. 

\section{Multi-Object Posterior Computation\protect\label{sec:MO-Posterior-Comp}}

Markov Chain Monte Carlo (MCMC) methods are essential for posterior
computation when analytical solutions are intractable due to complex
models or high-dimensional state/parameter spaces \cite{Geyer95},
\cite{Robert_Bayesian_Choice}. The key idea behind MCMC is to simulate
samples using a Markov chain---a stochastic process where each state/iterate
depends only on the previous one. The chain is designed so that its
long-run behavior (stationary distribution) matches the posterior.
Popular algorithms like Metropolis-Hastings, Gibbs sampling, and Hamiltonian
Monte Carlo implement this principle in different ways, trading-off
between efficiency and accuracy \cite{geman1984stochastic}. Once
the chain has \textquotedbl converged\textquotedbl , the samples
it produces can be used to approximate the desired distribution and
compute statistical estimates.

This section presents a Gibbs sampler for computing the multi-object
posterior. Gibbs sampling iteratively generates the value at each
coordinate (of the sample) in turn, given the current values at the
other coordinates. This approach offers several advantages. First,
it is conceptually simple and easy to implement, especially when the
full conditional distributions of the variables are known and easy
to sample from. Second, Gibbs sampling can be computationally efficient
in high-dimensional settings, as it updates one coordinate at a time
rather than proposing changes to the entire sample, with guaranteed
acceptance. Third, it avoids the need for choosing proposal distributions,
which simplifies the setup and reduces the risk of poor performance
due to suboptimal tuning. Finally, Gibbs sampling is guaranteed to
converge to the target distribution under mild regularity conditions.

\subsection{Gibbs Sampling for Multi-object Posterior Computation\protect\label{subsec:GS-MO-Posterior}}

The proposed multi-object posterior Gibbs sampler moves a given multi-object
trajectory sample $\boldsymbol{X}_{1:k}$ to a new sample $\acute{\boldsymbol{X}}_{1:k}$
by sequentially refreshing the values at each coordinate defined by
the time indices $1:k$, and labels in $\mathbb{L}_{k}$. Internally,
the proposed algorithm consists of two nested loops, an outer loop
which traverses all times $t=1:k$, and for each $t$, an inner loop
which traverses all labels $\ell\in\mathbb{L}_{t}$. 

The outer loop may traverse over time forwards or backwards, and for
each time $t$, moves the current multi-object state sample $\boldsymbol{X}_{t}$
to a new sample $\acute{\boldsymbol{X}}_{t}$. Due to the Markov property
of the system model, $\acute{\boldsymbol{X}}_{t}$ depends only on
the multi-object state samples at times $t-1,t,t+1$. If the loop
proceeds forward over time, $\acute{\boldsymbol{X}}_{t}$ depends
on $\acute{\boldsymbol{X}}_{t-1}$, $\boldsymbol{X}_{t}$, and $\boldsymbol{X}_{t+1}$;
and conversely if the loop proceeds backwards over time, $\acute{\boldsymbol{X}}_{t}$
depends on $\boldsymbol{X}_{t-1}$, $\boldsymbol{X}_{t}$, and $\acute{\boldsymbol{X}}_{t+1}$.

For each time $t$, the inner loop traverses all possible labels,
and for each label $\ell$, proposes a new single object sample $\acute{\boldsymbol{X}}_{t}^{\{\ell\}}$
to form the multi-object sample $\acute{\boldsymbol{X}}_{t}=\cup_{\ell\in\mathbb{L}_{t}}\acute{\boldsymbol{X}}_{t}^{\{\ell\}}$,
which takes into account survivals, deaths, and births. Each new single
object sample $\acute{\boldsymbol{X}}_{t}^{\{\ell\}}$ at time $t$
is either a singleton $\{(\acute{x},\ell)\}$ if a survival or a birth
is sampled, or an empty set $\emptyset$ otherwise. Hence, the conditional
distribution of each $\acute{\boldsymbol{X}}_{t}^{\{\ell\}}$ is a
Bernoulli whose existence probability and attribute density is a function
of the current measurement $z_{t},$ and the multi-object state samples
at times $t-1,t,t+1$. In fact, due to the conditional independence
of the attribute transition, we only need the attributes of the sampled
label at $t-1,t+1$. 

Intuitively, the label continuity property (from the standard transition
model) imposes the constraint that a label born before time $t+1$
and is still live at time $t+1$ , must be live at time $t$ (hence
has existence probability 1), and conversely a label born before time
$t-1$ and is not live at time $t-1$ cannot be live at time $t$
(hence has existence probability 0). If a label is sampled as live
at time $t$, its attribute is then sampled from the corresponding
attribute density, given its sampled states at times $t-1,t+1$ (due
to conditional independence of the attribute transition), and the
other members of the multi-object sample at time $t$. Sampling from
unnormalized single object densities can be implemented with any suitable
scheme such as Particle Flow (summarized in Appendix \ref{subsec:Particle-Flow}).

\subsection{Conditionals for the Gibbs Sampler\protect\label{subsec:cond-GS}}

The proposed Gibbs sampler described above relies on the ability to
sample from the conditional distribution of object $\ell$ at time
$t$, given the states of all other objects at time $t$ and the multi-object
states at all other times. This conditional Bernoulli RFS (because
object $\ell$ either exists or not exists) can be constructed from
the posterior $\boldsymbol{\pi}_{0:k}$ as 
\begin{alignat}{1}
\!\!\negthickspace\negthickspace\boldsymbol{\pi}_{t}^{(\ell)}(\boldsymbol{X}_{t}^{\{\ell\}}|\boldsymbol{X}_{0:t\texttt{-}1},\boldsymbol{X}_{t}^{\{\bar{\ell}\}\!},\boldsymbol{X}_{t\textup{\texttt{\textup{+}}}1:k})\nonumber \\
 & \negthickspace\negthickspace\negthickspace\negthickspace\negthickspace\negthickspace\negthickspace\negthickspace\negthickspace\negthickspace\negthickspace\negthickspace\negthickspace\negthickspace\negthickspace\negthickspace\negthickspace\negthickspace\negthickspace\negthickspace\negthickspace\negthickspace\negthickspace\negthickspace\negthickspace\negthickspace\negthickspace\negthickspace\negthickspace\negthickspace\negthickspace\negthickspace\negthickspace\negthickspace\negthickspace\negthickspace\negthickspace\negthickspace\triangleq\frac{\boldsymbol{\pi}_{t}(\boldsymbol{X}_{t}^{\{\ell\}\!}\uplus\boldsymbol{X}_{t}^{\{\bar{\ell}\}}|\boldsymbol{X}_{0:t\texttt{-}1},\boldsymbol{X}_{t\textup{\texttt{\textup{+}}}1:k})}{\int\!\boldsymbol{\pi}_{t}(\boldsymbol{Y}^{\{\ell\}\!}\uplus\boldsymbol{X}_{t}^{\{\bar{\ell}\}}|\boldsymbol{X}_{0:t\texttt{-}1},\boldsymbol{X}_{t\textup{\texttt{\textup{+}}}1:k})\delta\boldsymbol{Y}^{\{\ell\}}},\label{eq:conditional-l-t}
\end{alignat}
where $t\in\{1,...,k-1\}$, $\ell\in\mathbb{L}_{t}$, and
\begin{alignat}{1}
\!\!\!\boldsymbol{\pi}_{t}(\boldsymbol{X}_{t}|\boldsymbol{X}_{\!0:t\texttt{-}1},\boldsymbol{X}_{\!t\textup{\texttt{\textup{+}}}1:k}) & \triangleq\frac{\boldsymbol{\pi}_{0:k}(\boldsymbol{X}_{0:t\texttt{-}1},\boldsymbol{X}_{t},\boldsymbol{X}_{t\textup{\texttt{\textup{+}}}1:k})}{\int\!\boldsymbol{\pi}_{0:k}(\boldsymbol{X}_{\!0:t\texttt{-}1},\boldsymbol{Y}\!,\boldsymbol{X}_{\!t\textup{\texttt{\textup{+}}}1:k})\delta\boldsymbol{Y}},\!\label{eq:conditional-t}
\end{alignat}
is the conditional distribution of the multi-object state at time
$t$. The conditionals for $t=0,k$ are similarly defined.

Note that due to the Markov dynamics
\begin{alignat*}{1}
\boldsymbol{\pi}_{t}(\boldsymbol{X}_{t}|\boldsymbol{X}_{0:t\texttt{-}1},\boldsymbol{X}_{t\textup{\texttt{\textup{+}}}1:k})\\
 & \negthickspace\negthickspace\negthickspace\negthickspace\negthickspace\negthickspace\negthickspace\negthickspace\negthickspace\negthickspace\negthickspace\negthickspace\negthickspace\negthickspace\negthickspace\negthickspace\negthickspace\negthickspace\negthickspace\negthickspace\negthickspace\negthickspace\negthickspace\negthickspace\negthickspace\negthickspace\negthickspace\negthickspace\negthickspace\negthickspace\negthickspace\negthickspace\propto\boldsymbol{\pi}_{0}(\boldsymbol{X}_{0})\negthickspace\prod_{j\in\{1:k\}-\{t,t\texttt{\textup{+}}1\}}\negthickspace\boldsymbol{g}_{j}(z_{j}|\boldsymbol{X}_{j})\boldsymbol{f}_{\!j}(\boldsymbol{X}_{j}|\boldsymbol{X}_{j\texttt{-}1})\\
 & \negthickspace\negthickspace\negthickspace\negthickspace\negthickspace\negthickspace\negthickspace\negthickspace\negthickspace\negthickspace\negthickspace\negthickspace\negthickspace\negthickspace\negthickspace\negthickspace\negthickspace\negthickspace\negthickspace\negthickspace\negthickspace\negthickspace\negthickspace\negthickspace\negthickspace\negthickspace\negthickspace\negthickspace\negthickspace\negthickspace\negthickspace\negthickspace\times\boldsymbol{g}_{t\texttt{\textup{+}}1\!}(z_{t\texttt{\textup{+}}1}|\boldsymbol{X}_{t\texttt{\textup{+}}1})\boldsymbol{f}_{\!t\textup{\texttt{\textup{+}}}1\!}(\boldsymbol{X}_{t\texttt{\textup{+}}1}|\boldsymbol{X}_{t})\boldsymbol{g}_{t}(z_{t}|\boldsymbol{X}_{t})\boldsymbol{f}_{\!t}(\boldsymbol{X}_{t}|\boldsymbol{X}_{t\texttt{-}1})\\
 & \negthickspace\negthickspace\negthickspace\negthickspace\negthickspace\negthickspace\negthickspace\negthickspace\negthickspace\negthickspace\negthickspace\negthickspace\negthickspace\negthickspace\negthickspace\negthickspace\negthickspace\negthickspace\negthickspace\negthickspace\negthickspace\negthickspace\negthickspace\negthickspace\negthickspace\negthickspace\negthickspace\negthickspace\negthickspace\negthickspace\negthickspace\negthickspace\propto\boldsymbol{f}_{\!t\textup{\texttt{\textup{+}}}1\!}(\boldsymbol{X}_{t\texttt{\textup{+}}1}|\boldsymbol{X}_{t})\boldsymbol{f}_{\!t}(\boldsymbol{X}_{t}|\boldsymbol{X}_{t\texttt{-}1})\boldsymbol{g}_{t}(z_{t}|\boldsymbol{X}_{t}).
\end{alignat*}
This means the conditional $\boldsymbol{\pi}_{t}(\boldsymbol{X}_{t}|\boldsymbol{X}_{0:t\texttt{-}1},\boldsymbol{X}_{t\textup{\texttt{\textup{+}}}1:k})$
only depends on $\boldsymbol{X}_{t}$, the measurement $z_{t}$ and
its neighbouring multi-object states $\boldsymbol{X}_{t-1}$, $\boldsymbol{X}_{t+1}$.
Next, we establish that the product $\boldsymbol{f}_{\!t\textup{\texttt{\textup{+}}}1\!}(\boldsymbol{X}_{t\texttt{\textup{+}}1}|\boldsymbol{X}_{t})\boldsymbol{f}_{\!t}(\boldsymbol{X}_{t}|\boldsymbol{X}_{t\texttt{-}1})$
is an LMB in $\boldsymbol{X}_{t}$ (note for $t=k$, $\boldsymbol{f}_{\!t}(\boldsymbol{X}_{t}|\boldsymbol{X}_{t\texttt{-}1})$
is an LMB by default), and use this to determine an analytic characterization
of the Bernoulli conditional $\boldsymbol{\pi}_{t}^{(\ell)}(\boldsymbol{X}_{t}^{\{\ell\}}|\boldsymbol{X}_{0:t\texttt{-}1},\boldsymbol{X}_{t}^{\{\bar{\ell}\}},\boldsymbol{X}_{t\textup{\texttt{\textup{+}}}1:k})$.
The proofs are given in Appendix \ref{sec:Appendix-proof}.
\begin{prop}
\label{Prop:Transition-Prod-LMB}The product $\boldsymbol{f}_{\texttt{\textup{+}}}(\boldsymbol{X}_{\texttt{\textup{+}}}|\boldsymbol{X})\boldsymbol{f}(\boldsymbol{X}|\boldsymbol{X}_{\texttt{-}})$
of consecutive multi-object transition densities is an LMB in $\boldsymbol{X}$.
Specifically, \textup{}
\begin{alignat*}{1}
\boldsymbol{f}_{\texttt{\textup{+}}}(\boldsymbol{X}_{\texttt{\textup{+}}}|\boldsymbol{X})\boldsymbol{f}(\boldsymbol{X}|\boldsymbol{X}_{\texttt{-}})\propto & \mathbf{1}_{\mathcal{\mathcal{L}}(\boldsymbol{X}^{\mathbb{\mathbb{B}}})}^{\mathcal{L}(\boldsymbol{X}_{\texttt{\textup{+}}}^{\mathbb{\mathbb{L}}})\cap\mathbb{B}}\boldsymbol{\lambda}(\boldsymbol{X}^{\mathbb{\mathbb{B}}};r_{\boldsymbol{X}_{\texttt{\textup{+}}}},p_{\boldsymbol{X}_{\texttt{\textup{+}}}})\times\\
 & \mathbf{1}_{\mathcal{\mathcal{L}}(\boldsymbol{X}^{\mathbb{\mathbb{L_{\texttt{-}}}}})}^{\mathcal{L}(\boldsymbol{X}_{\texttt{\textup{+}}}^{\mathbb{\mathbb{L}}})\cap\mathbb{\mathbb{L_{\texttt{-}}}}}\boldsymbol{\lambda}(\boldsymbol{X}^{\mathbb{\mathbb{L_{\texttt{-}}}}};r_{\boldsymbol{X}_{\texttt{\textup{+}}},\boldsymbol{X}_{\texttt{-}}},p_{\boldsymbol{X}_{\texttt{\textup{+}}},\boldsymbol{X}_{\texttt{-}}}),
\end{alignat*}
where $\boldsymbol{\lambda}(\boldsymbol{X}^{\mathbb{\mathbb{B}}};r_{\boldsymbol{X}_{\texttt{\textup{+}}}},p_{\boldsymbol{X}_{\texttt{\textup{+}}}})$
and $\boldsymbol{\lambda}(\boldsymbol{X}^{\mathbb{\mathbb{L_{\texttt{-}}}}};r_{\boldsymbol{X}_{\texttt{\textup{+}}},\boldsymbol{X}_{\texttt{-}}},p_{\boldsymbol{X}_{\texttt{\textup{+}}},\boldsymbol{X}_{\texttt{-}}})$
are LMBs, with
\begin{alignat*}{1}
{\normalcolor {\normalcolor {\normalcolor {\normalcolor r_{\boldsymbol{X}_{\texttt{\textup{+}}}}(\ell}\mathclose{\normalcolor )}}}} & {\normalcolor {\normalcolor \mathrel{\normalcolor =}}}\left\lceil \!\frac{\langle\overset{\shortrightarrow}{f}(\boldsymbol{X}_{\texttt{\textup{+}}}^{\mathbb{\mathbb{L}}};\cdot)b(\cdot)\rangle(\ell)}{Q_{B}(\ell)}\!\right\rfloor ^{1-|\boldsymbol{X}_{\texttt{\textup{+}}}^{\{\ell\}}|},\\
{\normalcolor {\normalcolor p_{\boldsymbol{X}_{\texttt{\textup{+}}}}(x,\ell}\mathclose{\normalcolor )}} & {\normalcolor \mathrel{\normalcolor =}}\frac{\overset{\shortrightarrow}{f}(\boldsymbol{X}_{\texttt{\textup{+}}}^{\mathbb{\mathbb{L}}};x,\ell)b(x,\ell)}{\langle\overset{\shortrightarrow}{f}(\boldsymbol{X}_{\texttt{\textup{+}}}^{\mathbb{\mathbb{L}}};\cdot)b(\cdot)\rangle(\ell)},\\
{\normalcolor {\normalcolor r_{\boldsymbol{X}_{\texttt{\textup{+}}},\boldsymbol{X}_{\texttt{-}}}(\ell}\mathclose{\normalcolor )}} & {\normalcolor \mathrel{\normalcolor =}}\left\lceil \!\frac{\langle\overset{\shortrightarrow}{f}(\boldsymbol{X}_{\texttt{\textup{+}}}^{\mathbb{\mathbb{L}}};\cdot)\overset{\shortleftarrow}{f}(\cdot;\boldsymbol{X}_{\texttt{-}})\rangle(\ell)}{Q_{S}^{\boldsymbol{X}_{\texttt{-}}^{\{\ell\}}}}\!\right\rfloor ^{1-|\boldsymbol{X}_{\texttt{\textup{+}}}^{\{\ell\}}|},\\
{\normalcolor p_{\boldsymbol{X}_{\texttt{\textup{+}}},\boldsymbol{X}_{\texttt{-}}}(x,\ell}\mathclose{\normalcolor )} & \mathrel{\normalcolor =}\frac{\overset{\shortrightarrow}{f}(\boldsymbol{X}_{\texttt{\textup{+}}}^{\mathbb{\mathbb{L}}};x,\ell)\overset{\shortleftarrow}{f}(x,\ell;\boldsymbol{X}_{\texttt{-}})}{\langle\overset{\shortrightarrow}{f}(\boldsymbol{X}_{\texttt{\textup{+}}}^{\mathbb{\mathbb{L}}};\cdot)\overset{\shortleftarrow}{f}(\cdot;\boldsymbol{X}_{\texttt{-}})\rangle(\ell)}.
\end{alignat*}
\end{prop}
The resulting LMB is the product of two independent LMBs with disjoint
label sets: a birth LMB and a survival LMB. For a given set $\boldsymbol{X}_{\texttt{\textup{+}}}^{\mathbb{\mathbb{L}}}$
of surviving objects at the time $t+1$, the set inclusions $\mathbf{1}_{\mathcal{\mathcal{L}}(\boldsymbol{X}^{\mathbb{\mathbb{B}}})}^{\mathcal{L}(\boldsymbol{X}_{\texttt{\textup{+}}}^{\mathbb{\mathbb{L}}})\cap\mathbb{B}}$
and $\mathbf{1}_{\mathcal{\mathcal{L}}(\boldsymbol{X}^{\mathbb{\mathbb{L_{\texttt{-}}}}})}^{\mathcal{L}(\boldsymbol{X}_{\texttt{\textup{+}}}^{\mathbb{\mathbb{L}}})\cap\mathbb{\mathbb{L_{\texttt{-}}}}}$
ensure that the labels of the arguments $\boldsymbol{X}^{\mathbb{\mathbb{B}}}$
and $\boldsymbol{X}^{\mathbb{\mathbb{L_{\texttt{-}}}}}$ of the respective
LMBs include the labels of $\boldsymbol{X}_{\texttt{\textup{+}}}^{\mathbb{\mathbb{L}}}$
for label continuity. In addition, the cardinality $|\boldsymbol{X}_{\texttt{\textup{+}}}^{\{\ell\}}|$
in the existence probabilities ensures that any live label (that survives)
to time $t+1$ must have existence probability 1 (and hence also be
live) at time $t$.\textcolor{red}{{} }It is implicit from the term
$Q_{S}^{\boldsymbol{X}_{\texttt{-}}^{\{\ell\}}}$ in ${\color{lime}{\normalcolor {\normalcolor {\normalcolor r_{\boldsymbol{X}_{\texttt{\textup{+}}}^{\mathbb{\mathbb{L}}},\boldsymbol{X}_{\texttt{-}}}(\ell}\mathclose{\normalcolor )}}}}$
that the existence probabilities of survival LMB are not defined for
labels not contained in $\boldsymbol{\boldsymbol{X}_{\texttt{-}}}$.
This means any label that was not present at time $t-1$ cannot appear
at time $t$, unless it is generated as a birth. 
\begin{prop}
\label{Prop:Bernoulli-Conditional}Suppose that the multi-object density
$\boldsymbol{\pi}$ is proportional to the product of an LMB $\boldsymbol{\lambda}(\cdot;r,p)$
and an arbitrary function $\phi$ (of the multi-object state), i.e.,
\begin{alignat*}{1}
\boldsymbol{\pi}(\boldsymbol{X}) & \propto\phi(\boldsymbol{X})\boldsymbol{\lambda}(\boldsymbol{X};r,p).
\end{alignat*}
Then, the Bernoulli conditional distribution
\begin{alignat*}{1}
\boldsymbol{\pi}(\boldsymbol{X}^{\{\ell\}}|\boldsymbol{X}^{\{\bar{\ell}\}}) & \triangleq\frac{\boldsymbol{\pi}(\boldsymbol{X}^{\{\ell\}}\uplus\boldsymbol{X}^{\{\bar{\ell}\}})}{\int\boldsymbol{\pi}(\boldsymbol{Y}^{\{\ell\}}\uplus\boldsymbol{X}^{\{\bar{\ell}\}})\delta\boldsymbol{Y}^{\{\ell\}}},
\end{alignat*}
has existence probability and attribute density:
\begin{alignat*}{1}
r(\ell|\boldsymbol{X}^{\{\bar{\ell}\}}) & =\left\lceil \!\frac{r(\ell)\bigl\langle p(\cdot,\ell),\phi(\{(\cdot,\ell)\}\uplus\boldsymbol{X}^{\{\bar{\ell}\}})\bigr\rangle}{\left[1-r(\ell)\right]\phi(\boldsymbol{X}^{\{\bar{\ell}\}})}\!\right\rfloor \!,\\
p(x,\ell|\boldsymbol{X}^{\{\bar{\ell}\}}) & =\frac{p(x,\ell)\phi(\{(x,\ell)\}\uplus\boldsymbol{X}^{\{\bar{\ell}\}})}{\bigl\langle p(\cdot,\ell),\phi(\{(\cdot,\ell)\}\uplus\boldsymbol{X}^{\{\bar{\ell}\}})\bigr\rangle}.
\end{alignat*}
\end{prop}
Proposition \ref{Prop:Transition-Prod-LMB} asserts that the multi-object
conditional $\boldsymbol{\pi}_{t}(\boldsymbol{X}_{t}|\boldsymbol{X}_{0:t\texttt{-}1},\boldsymbol{X}_{t\textup{\texttt{\textup{+}}}1:k})$
is proportional to the product of an LMB and the multi-object likelihood.
Consequently, Proposition \ref{Prop:Bernoulli-Conditional} can be
invoked to determine analytic expressions for the (Bernoulli) conditional
of object $\ell$ at time $t$, summarized in Corollary \ref{Cor:Bernoulli-Conditional}.
This conditional is straightforward to evaluate and sample from, thereby
enabling efficient Gibbs updates. 

\textit{Remark}: In general, if the posterior $\boldsymbol{\pi}(\cdot)$
is proportional to some function $\boldsymbol{\varphi}(\cdot)$, then
the conditional given by $\boldsymbol{\pi}(\boldsymbol{X}{}^{\{\ell\}}|\boldsymbol{X}^{\{\bar{\ell}\}})$
is Bernoulli with
\begin{alignat*}{1}
r(\ell|\boldsymbol{X}^{\{\bar{\ell}\}}) & =\left\lceil \!\frac{\boldsymbol{\varphi}(\{(\cdot,\ell)\}\uplus\boldsymbol{X}^{\{\bar{\ell}\}})\bigr\rangle}{\boldsymbol{\varphi}(\boldsymbol{X}^{\{\bar{\ell}\}})}\!\right\rfloor \!,\\
p(x,\ell|\boldsymbol{X}^{\{\bar{\ell}\}}) & =\frac{\boldsymbol{\varphi}(\{(x,\ell)\}\uplus\boldsymbol{X}^{\{\bar{\ell}\}})}{\bigl\langle1,\boldsymbol{\varphi}(\{(\cdot,\ell)\}\uplus\boldsymbol{X}^{\{\bar{\ell}\}})\bigr\rangle}.
\end{alignat*}
Thus, a generic transition density can also be considered, but its
tractability depends on the specific form (for which the LMB case
above results in simple closed form).
\begin{cor}
\label{Cor:Bernoulli-Conditional}The existence probability and attribute
density of the (Bernoulli) conditional $\boldsymbol{\pi}_{t}^{(\ell)}(\boldsymbol{X}_{t}^{\{\ell\}}|\boldsymbol{X}_{0:t\texttt{-}1},\boldsymbol{X}_{t}^{\{\bar{\ell}\}},\boldsymbol{X}_{t\textup{\texttt{\textup{+}}}1:k})$
of object $\ell$ at time $t$, is given by 
\begin{alignat*}{1}
r_{t}^{(\ell)}(\boldsymbol{X}^{\{\bar{\ell}\}}\!,\boldsymbol{X}_{\!\pm}) & =\left\lceil \!\frac{{\scriptstyle r(\ell|\boldsymbol{X}_{\!\pm})\bigl\langle p(\cdot,\ell|\boldsymbol{X}_{\!\pm}),\boldsymbol{g}(z|\{(\cdot,\ell)\}\uplus\boldsymbol{X}^{\{\bar{\ell}\}})\bigr\rangle}}{{\scriptstyle \left[1-r(\ell|\boldsymbol{X}_{\pm})\right]\boldsymbol{g}(z|\boldsymbol{X}^{\{\bar{\ell}\}})}}\!\right\rfloor \!,\\
p_{t}^{(\ell)}(x|\boldsymbol{X}^{\{\bar{\ell}\}}\!,\boldsymbol{X}_{\pm}) & =\frac{{\scriptstyle p(x,\ell|\boldsymbol{X}_{\pm})\boldsymbol{g}(z|\{(x,\ell)\}\uplus\boldsymbol{X}^{\{\bar{\ell}\}})}}{{\scriptstyle \bigl\langle p(\cdot,\ell|\boldsymbol{X}_{\pm}),\boldsymbol{g}(z|\{(\cdot,\ell)\}\uplus\boldsymbol{X}^{\{\bar{\ell}\}})\bigr\rangle}},
\end{alignat*}
where
\begin{alignat*}{1}
r(\ell|\boldsymbol{X}_{\pm}) & =\begin{cases}
\!\!\!\begin{array}{ll}
\left\lceil \!\frac{\langle\overset{\shortrightarrow}{f}(\boldsymbol{X}_{\texttt{\textup{+}}}^{\mathbb{\mathbb{L}}};\cdot,\ell),b(\cdot,\ell)\rangle}{Q_{B}(\ell)}\!\right\rfloor ^{1-|\boldsymbol{X}_{\texttt{\textup{+}}}^{\{\ell\}}|}, & \!\!\!\ell\in\mathbb{B}\\
\left\lceil \!\frac{\langle\overset{\shortrightarrow}{f}(\boldsymbol{X}_{\texttt{\textup{+}}}^{\mathbb{\mathbb{L}}};\cdot,\ell),\overset{\shortleftarrow}{f}(\cdot,\ell;\boldsymbol{X}_{\texttt{-}})\rangle}{Q_{S}^{\boldsymbol{X}_{\texttt{-}}^{\{\ell\}}}}\!\right\rfloor ^{1-|\boldsymbol{X}_{\texttt{\textup{+}}}^{\{\ell\}}|}, & \!\!\!\ell\in\mathbb{L}_{\texttt{-}}
\end{array}\!\!\!\!\end{cases}\\
p(x,\ell|\boldsymbol{X}_{\pm}) & =\begin{cases}
\!\!\!\begin{array}{ll}
\frac{\overset{\shortrightarrow}{f}(\boldsymbol{X}_{\texttt{\textup{+}}}^{\mathbb{\mathbb{L}}};x,\ell)b(x,\ell)}{\langle\overset{\shortrightarrow}{f}(\boldsymbol{X}_{\texttt{\textup{+}}}^{\mathbb{\mathbb{L}}};\cdot,\ell),b(\cdot,\ell)\rangle}, & \!\!\!\ell\in\mathbb{B}\\
\frac{\overset{\shortrightarrow}{f}(\boldsymbol{X}_{\texttt{\textup{+}}}^{\mathbb{\mathbb{L}}};x,\ell)\overset{\shortleftarrow}{f}(x,\ell;\boldsymbol{X}_{\texttt{-}})}{\langle\overset{\shortrightarrow}{f}(\boldsymbol{X}_{\texttt{\textup{+}}}^{\mathbb{\mathbb{L}}};\cdot,\ell),\overset{\shortleftarrow}{f}(\cdot,\ell;\boldsymbol{X}_{\texttt{-}})\rangle}, & \!\!\!\ell\in\mathbb{L}_{\texttt{-}}
\end{array}\!\!\!.\!\end{cases}
\end{alignat*}
 
\end{cor}

\subsection{Multi-Object Posterior Gibbs Sampler\protect\label{subsec:MO-Posterior-GS}}

Having derived explicit conditionals for the multi-object posterior,
we now proceed to the pseudo code for the resulting Gibbs sampler.
Algorithm 1 describes the Markov chain for a given initial multi-object
trajectory. The overall flow of execution follows the qualitative
description given in Subsection \ref{subsec:GS-MO-Posterior} and
uses the conditional distribution derived in Subsection \ref{subsec:cond-GS}.
Algorithm 1a is the outer loop which traverses forwards or backwards
over time to sample a multi-object trajectory. Algorithm 1b is the
inner loop which traverses over all possible labels for a given time
to sample a multi-object state. The core of the Gibbs sampler is the
LMB conditional, which, for a given label, proposes a new single object
sample, either as an empty set or singleton. 

In Algorithm 1a, the evaluation of the loop iterations for times $t=0$
and $t=k$ are treated as special cases since it is implicit that
the inputs \textsf{$\boldsymbol{X}_{t-1}^{(i)}$} and \textsf{$\boldsymbol{X}_{t+1}^{(i)}$}
respectively are ignored. Thus in Algorithm 1b, the case of $t=0$
reduces to 
\begin{eqnarray*}
p(\cdot,\ell|\boldsymbol{X}_{\pm}) & :\propto & \begin{cases}
\!\!\!\begin{array}{ll}
\overset{\shortrightarrow}{f}(\boldsymbol{X}_{\texttt{\textup{+}}}^{\mathbb{\mathbb{L}}};\cdot,\ell)b(\cdot,\ell) & \!\!\!\ell\in\mathbb{B}\\
\overset{\shortrightarrow}{f}(\boldsymbol{X}_{\texttt{\textup{+}}}^{\mathbb{\mathbb{L}}};\cdot,\ell) & \!\!\!\ell\in\mathbb{L}_{\texttt{-}}
\end{array}\!\!\!\!\end{cases}\\
r(\ell|\boldsymbol{X}_{\pm}) & := & \begin{cases}
\!\!\!\begin{array}{ll}
\left\lceil \!\frac{\langle\overset{\shortrightarrow}{f}(\boldsymbol{X}_{\texttt{\textup{+}}}^{\mathbb{\mathbb{L}}};\cdot,\ell),b(\cdot,\ell)\rangle}{Q_{B}(\ell)}\!\right\rfloor ^{1-|\boldsymbol{X}_{\texttt{\textup{+}}}^{\{\ell\}}|} & \!\!\!\ell\in\mathbb{B}\\
\left\lceil \!\frac{\langle\overset{\shortrightarrow}{f}(\boldsymbol{X}_{\texttt{\textup{+}}}^{\mathbb{\mathbb{L}}};\cdot,\ell),1\rangle}{Q_{S}^{\boldsymbol{X}_{\texttt{-}}^{\{\ell\}}}}\!\right\rfloor ^{1-|\boldsymbol{X}_{\texttt{\textup{+}}}^{\{\ell\}}|} & \!\!\!\ell\in\mathbb{L}_{\texttt{-}}
\end{array}\!\!\!\!\end{cases}
\end{eqnarray*}
and the case of $t=k$ reduces to

\begin{eqnarray*}
p(\cdot,\ell|\boldsymbol{X}_{\pm}) & :\propto & \begin{cases}
\!\!\!\begin{array}{ll}
b(\cdot,\ell) & \!\!\!\ell\in\mathbb{B}\\
\overset{\shortleftarrow}{f}(\cdot,\ell;\boldsymbol{X}_{\texttt{-}}) & \!\!\!\ell\in\mathbb{L}_{\texttt{-}}
\end{array}\!\!\!\!\end{cases}\\
r(\ell|\boldsymbol{X}_{\pm}) & := & \begin{cases}
\!\!\!\begin{array}{ll}
P_{B}(\ell) & \!\!\!\ell\in\mathbb{B}\\
\left\lceil \!\frac{\langle1,\overset{\shortleftarrow}{f}(\cdot,\ell;\boldsymbol{X}_{\texttt{-}})\rangle}{Q_{S}^{\boldsymbol{X}_{\texttt{-}}^{\{\ell\}}}}\!\right\rfloor  & \!\!\!\ell\in\mathbb{L}_{\texttt{-}}
\end{array}\!\!\!\!\end{cases}
\end{eqnarray*}
Note the slight abuse of notation in the conceptual overloading of
the variables, where $\mathbb{L}$ effectively restricts the main
loop to active labels (hence carries a different meaning from Section
\ref{sec:Background} Subsection \ref{subsec:cond-GS}).

The multi-object trajectory Gibbs sampler requires initialization
with any valid multi-object trajectory, including the empty trajectory
which is equivalent to the sequence of empty sets. A typically more
efficient scheme is to employ a so-called factor sampler which is
analogous to recursive filtering while accumulating previous time
samples. The factor sampler can be implemented as special case of
the Gibbs sampler for an expanding time window. The resulting initialization
routine is described in Algorithm 2 where the outer loop runs forward
in time and the inner loop is akin to a filtering step. To save repetition
in the description of Algorithm 2, the core sampling routine reuses
Algorithm 1b by substituting an empty argument for $\boldsymbol{X}_{t+1}^{(i)}$,
which similarly implements the special case $t=k$ given above.

The proposed multi-object sampler is a single\nobreakdash-site Gibbs
sampler with a systematic scan that updates coordinates in a fixed
order, drawing exactly from the corresponding full conditional at
each step. The composed kernel preserves the stationary distribution
$\boldsymbol{\pi}_{0:k}$ and, under standard conditions (irreducibility,
aperiodicity, proper conditionals), the chain is Harris ergodic; with
mild drift--minorization conditions, geometric ergodicity follows
(\cite{Tierney1994,RobertsSmith1994,RobertCasella2004,MeynTweedie2009}).
Note that a random\nobreakdash-scan variant that selects coordinate
$(\ell,t)$ with probability $p_{\ell,t}$, is $\boldsymbol{\pi}_{0:k}$-reversible
and inherits the same convergence properties (\cite{Tierney1994,Liu2001,RobertCasella2004}).
Systematic-scan is simpler to implement and can mix well in practice,
but in theory reversibility may be lost. On the other hand, random-scan
offers reversibility of the overall chain, but can be slow-mixing
in practice.

\begin{figure}[htbp]
\begin{tabular}{>{\raggedright}p{8.6cm}}
\toprule 
\textbf{\textit{Algorithm 1a. MOT Gibbs Sampler (FW/BW)}}\tabularnewline
\midrule 
\textsf{Global:} $\{P_{B}(\cdot),b(\cdot,\cdot)\},P_{S}(\cdot,\cdot),\overset{\shortrightarrow}{f}(\cdot|\cdot),\overset{\shortleftarrow}{f}(\cdot|\cdot),z_{1:k},g(\cdot|\cdot)$\tabularnewline
\textsf{Input:\ \ \ \  -}\tabularnewline
\textsf{Output:} $\{\boldsymbol{X}_{0:k}\}_{i=1}^{I_{max}}$\tabularnewline
\midrule
$\boldsymbol{X}_{0:k}^{(0)}:=$\textsf{Initialize$()$}\tabularnewline
\textsf{for} $i=1:I_{max}$\tabularnewline
\thickspace{}\thickspace{}\textsf{for} $t=0:1:k$ \textsf{OR} $t=k:-1:0$\tabularnewline
\thickspace{}\thickspace{}\thickspace{}\thickspace{}$\boldsymbol{X}_{t}^{(i)}:=$\textsf{MOSGibbsSampler}$(\boldsymbol{X}_{t-1}^{(i)},\boldsymbol{X}_{t}^{(i-1)},\boldsymbol{X}_{t+1}^{(i-1)})$\tabularnewline
\thickspace{}\thickspace{}\thickspace{}\thickspace{}\thickspace{}\thickspace{}\thickspace{}\thickspace{}\thickspace{}\thickspace{}\thickspace{}\thickspace{}\thickspace{}\thickspace{}\thickspace{}\thickspace{}\thickspace{}\thickspace{}\thickspace{}\thickspace{}\thickspace{}\thickspace{}
\thickspace{}\thickspace{} \thickspace{}\thickspace{} \textsf{OR}\tabularnewline
\thickspace{}\thickspace{}\thickspace{}\thickspace{}$\boldsymbol{X}_{t}^{(i)}:=$\textsf{MOSGibbsSampler}$(\boldsymbol{X}_{t-1}^{(i-1)},\boldsymbol{X}_{t}^{(i-1)},\boldsymbol{X}_{t+1}^{(i)})$\tabularnewline
\thickspace{}\thickspace{}\textsf{end}\tabularnewline
\thickspace{}\thickspace{}$\boldsymbol{X}_{0:k}^{(i)}:=\left[\boldsymbol{X}_{0}^{(i)},...,\boldsymbol{X}_{k}^{(i)}\right]$\tabularnewline
\textsf{end}\tabularnewline
\bottomrule
\end{tabular}

\medskip{}

\begin{tabular}{>{\raggedright}p{8.6cm}}
\toprule 
\textbf{\textit{Algorithm 1b. MOS Gibbs Sampler}}\tabularnewline
\midrule 
\textsf{Input:} \textsf{\ \ \ }$\boldsymbol{X}_{-},\boldsymbol{X},\boldsymbol{X}_{+}$\tabularnewline
\textsf{Output:} $\boldsymbol{X}^{\prime}$\tabularnewline
\midrule
$\mathbb{L}:=\mathcal{L}(\boldsymbol{X}_{-})\cup\mathbb{B}$\tabularnewline
\textsf{for} $\ell_{}\in\mathbb{L}$\tabularnewline
\thickspace{}\thickspace{}$\boldsymbol{X}^{\{\bar{\ell}\}}:=\boldsymbol{X}-\boldsymbol{X}^{\{\ell\}}$\tabularnewline
\thickspace{}\thickspace{}$p(\cdot,\ell|\boldsymbol{X}_{\pm}):\propto\begin{cases}
\!\!\!\begin{array}{ll}
\overset{\shortrightarrow}{f}(\boldsymbol{X}_{\texttt{\textup{+}}}^{\mathbb{\mathbb{L}}};\cdot,\ell)b(\cdot,\ell) & \!\!\!\ell\in\mathbb{B}\\
\overset{\shortrightarrow}{f}(\boldsymbol{X}_{\texttt{\textup{+}}}^{\mathbb{\mathbb{L}}};\cdot,\ell)\overset{\shortleftarrow}{f}(\cdot,\ell;\boldsymbol{X}_{\texttt{-}}) & \!\!\!\ell\in\mathbb{L}_{\texttt{-}}
\end{array}\!\!\!\!\end{cases}$\tabularnewline
\thickspace{}\thickspace{}$r(\ell|\boldsymbol{X}_{\pm}):=\begin{cases}
\!\!\!\begin{array}{ll}
\left\lceil \!\frac{\langle\overset{\shortrightarrow}{f}(\boldsymbol{X}_{\texttt{\textup{+}}}^{\mathbb{\mathbb{L}}};\cdot,\ell),b(\cdot,\ell)\rangle}{Q_{B}(\ell)}\!\right\rfloor ^{1-|\boldsymbol{X}_{\texttt{\textup{+}}}^{\{\ell\}}|} & \!\!\!\ell\in\mathbb{B}\\
\left\lceil \!\frac{\langle\overset{\shortrightarrow}{f}(\boldsymbol{X}_{\texttt{\textup{+}}}^{\mathbb{\mathbb{L}}};\cdot,\ell),\overset{\shortleftarrow}{f}(\cdot,\ell;\boldsymbol{X}_{\texttt{-}})\rangle}{Q_{S}^{\boldsymbol{X}_{\texttt{-}}^{\{\ell\}}}}\!\right\rfloor ^{1-|\boldsymbol{X}_{\texttt{\textup{+}}}^{\{\ell\}}|} & \!\!\!\ell\in\mathbb{L}_{\texttt{-}}
\end{array}\!\!\!\!\end{cases}$\tabularnewline
\thickspace{}\thickspace{}$p^{(\ell)}(\cdot|\boldsymbol{X}^{\{\bar{\ell}\}}\!,\boldsymbol{X}_{\pm}):\propto p(\cdot,\ell|\boldsymbol{X}_{\pm})\boldsymbol{g}(z|\{(\cdot,\ell)\}\uplus\boldsymbol{X}^{\{\bar{\ell}\}})$

\thickspace{}\thickspace{}\thickspace{}\thickspace{}\thickspace{}\thickspace{}\thickspace{}\thickspace{}\thickspace{}\thickspace{}\thickspace{}\thickspace{}\thickspace{}\thickspace{}\thickspace{}\thickspace{}\thickspace{}\thickspace{}\thickspace{}\thickspace{}\thickspace{}\thickspace{}\thickspace{}\thickspace{}\thickspace{}\thickspace{}\thickspace{}\thickspace{}\thickspace{}\thickspace{}\thickspace{}\thickspace{}\thickspace{}\textsf{{[}via
PF alg{]}}\tabularnewline
\thickspace{}\thickspace{}$r^{(\ell)}(\boldsymbol{X}^{\{\bar{\ell}\}}\!,\boldsymbol{X}_{\!\pm}):=\left\lceil \!\frac{{\scriptstyle r(\ell|\boldsymbol{X}_{\!\pm})\bigl\langle p(\cdot,\ell|\boldsymbol{X}_{\!\pm}),\boldsymbol{g}(z|\{(\cdot,\ell)\}\uplus\boldsymbol{X}^{\{\bar{\ell}\}})\bigr\rangle}}{{\scriptstyle \left[1-r(\ell|\boldsymbol{X}_{\pm})\right]\boldsymbol{g}(z|\boldsymbol{X}^{\{\bar{\ell}\}})}}\!\right\rfloor $\tabularnewline
\thickspace{}\thickspace{}$\LyXZeroWidthSpace\boldsymbol{X}^{\{\ell\}}:=\emptyset$\tabularnewline
\thickspace{}\thickspace{}\textsf{if} \textsf{Uniform}$(0,1)<r^{(\ell)}(\boldsymbol{X}^{\{\bar{\ell}\}}\!,\boldsymbol{X}_{\!\pm})$\tabularnewline
\thickspace{}\thickspace{}\thickspace{}\thickspace{}$x\sim p^{(\ell)}(\cdot|\boldsymbol{X}^{\{\bar{\ell}\}}\!,\boldsymbol{X}_{\pm})$\tabularnewline
\thickspace{}\thickspace{}\thickspace{}\thickspace{}$\LyXZeroWidthSpace\boldsymbol{X}^{\{\ell\}}:=\{(x,\ell)\}$\tabularnewline
\thickspace{}\thickspace{}\textsf{end}\tabularnewline
\thickspace{}\thickspace{}$\LyXZeroWidthSpace\boldsymbol{X}:=\boldsymbol{X}^{\{\ell\}}\cup\boldsymbol{X}^{\{\bar{\ell}\}}$\tabularnewline
\textsf{end}\tabularnewline
$\boldsymbol{X}_{}^{\prime}:=\boldsymbol{X}$\tabularnewline
\bottomrule
\end{tabular}
\end{figure}

\begin{figure}[htbp]
\begin{tabular}{>{\raggedright}p{8.6cm}}
\toprule 
\textbf{\textit{Algorithm 2. MOT Factor Sampler}}\tabularnewline
\midrule 
\textsf{Global:} $\{P_{B}(\cdot),b(\cdot,\cdot)\},P_{S}(\cdot,\cdot),\overset{\shortleftarrow}{f}(\cdot|\cdot),z_{1:k},g(\cdot|\cdot)$\tabularnewline
\textsf{Input: \ \ \ -}\tabularnewline
\textsf{Output:} $\boldsymbol{X}_{0:k}=\left[\boldsymbol{X}_{0},...,\boldsymbol{X}_{k}\right]$\tabularnewline
\midrule
$\boldsymbol{X}_{0}:=\emptyset$\tabularnewline
\textsf{for} $t=1:k$\tabularnewline
\thickspace{}\thickspace{}$\boldsymbol{X}_{t}:=\emptyset$\tabularnewline
\thickspace{}\thickspace{}\textsf{for} $i=1:I_{max}$\tabularnewline
\thickspace{}\thickspace{}\thickspace{}\thickspace{}$\boldsymbol{X}_{t}:=$\textsf{MOSGibbsSampler}$(\boldsymbol{X}_{t-1},\boldsymbol{X}_{t},[])$\tabularnewline
\thickspace{}\thickspace{}\textsf{end}\tabularnewline
\textsf{end}\tabularnewline
$\boldsymbol{X}_{0:k}:=\left[\boldsymbol{X}_{0},...,\boldsymbol{X}_{k}\right]$\tabularnewline
\bottomrule
\end{tabular}

\medskip{}
\end{figure}

The proposed Gibbs sampler requires sampling from an unnormalized
single-object posterior, given a notional prior and likelihood. Since
a multi-object trajectory is composed of many correlated single-object
trajectories, high efficiency at the single-object level is crucial
for overall tractability. Any method capable of generating single-trajectory
particles can be used. For the case study in the next section, we
use Particle Flow (also known as log-homotopy) to exploit the underlying
smoothness of the single-object model, and provide a more accurate
posterior approximation with a relatively small number of samples
\cite{DaumHuang2007LogHomotopy}, \cite{DaiDaum2023StochasticPF}.
Details are given in Appendix \ref{subsec:Particle-Flow}.

\section{Superpositional Measurement Case Study\protect\label{sec:Superpositional}}

As an example application of the proposed multi-object posterior Gibbs
sampler, this section presents a case study on multi-object estimation
with superpositional measurements--each measurement is a field whose
intensity is a sum of contributions from multiple objects. In many
sensing modalities (e.g., low\nobreakdash-SNR radar/sonar power maps,
EO/IR image intensities, radio\nobreakdash-tomographic fields, atmospheric
concentration maps), the observation at each spatial sample (pixel/beam/bin)
is a superposition of contributions from all objects within the sensor\textquoteright s
footprint plus background/noise \cite{MahlerSCPHD09,NCM2013Computationally,PapiKim-15,LWLGZ2020Robust}.

\subsection{Experiment Settings}

The measurement is an image $z_{t}=(z_{t,m})_{m=1}^{M}$ of $M$
cells (pixels, range--Doppler cells, etc.), that are conditionally
independent given the multi-object state, i.e,
\begin{alignat*}{1}
\boldsymbol{g}_{t}(z_{t}|\boldsymbol{X}_{t}) & =\prod_{m=1}^{M}\boldsymbol{g}_{t,m}(z_{t,m}|\boldsymbol{X}_{t}).
\end{alignat*}
In this example we use
\begin{alignat*}{1}
\boldsymbol{g}_{t,m}(z_{t,m}|\boldsymbol{X}_{t}) & =\mathcal{N}\left(z_{t,m}-\sum_{\boldsymbol{x}\in\boldsymbol{X}_{t}}A_{m}(\boldsymbol{x});0,1\right),
\end{alignat*}
where the point spread function is given by
\begin{alignat*}{1}
A_{m}(\boldsymbol{x}) & =\frac{\Delta_{x}\Delta_{y}I_{0}}{2\pi\sigma_{s}^{2}}e^{\left(-\frac{(p_{x}-m_{x})^{2}+(p_{y}-m_{y}){}^{2}}{2\sigma_{s}^{2}}\right)}
\end{alignat*}
with $(p_{x},p_{y})$ denoting the position of state $\boldsymbol{x}$,
$(m_{x},m_{y})$ the position of pixel $m$, $\Delta_{x}$ and $\Delta_{y}$
the respective horizontal and vertical pixel width, $I_{0}$ the source
amplitude, and $\sigma_{s}^{2}$ the blurring factor. 
\begin{figure}[t]
\begin{centering}
\resizebox{80mm}{!}{\includegraphics[clip]{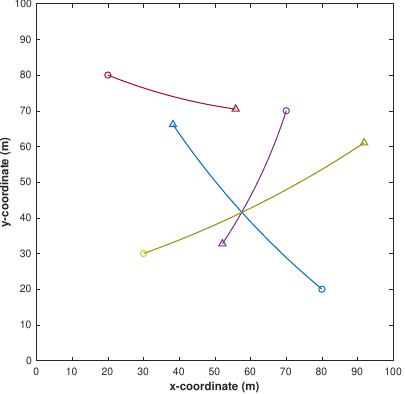}}
\par\end{centering}
\caption{\protect\label{fig:gt} Ground truth: Four trajectories, with starting
and end points annotated by '${\scriptstyle \ocircle}$' and '${\scriptstyle {\textstyle \vartriangle}}$'.
These objects appear and disappear at different times, with three
of them crossing at $t=60s$. }
\vspace{0mm}
\end{figure}

We consider a $100m\times100m$ surveillance region, which is observed
as a $100\times100$ pixel image (i.e., $M=10000$ pixels) with $\Delta_{x}=\Delta_{y}=1m$
and $\sigma_{s}^{2}=1$.\textcolor{violet}{{} }We also examine a range
of source amplitudes $I_{0}=15,20,25,30,35$, which effectively varies
the signal-to-noise ratio.
\begin{figure*}[t]
\begin{centering}
\resizebox{160mm}{!}{\includegraphics[clip]{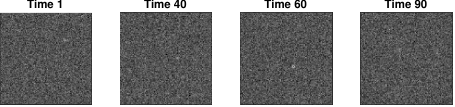}}
\par\end{centering}
\caption{\protect\label{fig:meas}Sample superpositional image observations
at times $k=1s,40s,60s,90s$, for source level $I_{0}=15$.}
\vspace{0mm}
\end{figure*}

An object state is a 5D vector of planar position, velocity and turn
rate $x_{t}=[p_{t,x},p_{t,y},v_{t,x},v_{t,y},\omega_{t}]^{T}$. Each
object follows a nearly constant turn model given by
\begin{eqnarray*}
x_{t} & = & F(\omega_{t-1})x_{t-1}+Gw_{t-1},\\
\omega_{t} & = & \omega_{t-1}+\iota u_{t-1},
\end{eqnarray*}
where 
\begin{alignat*}{1}
F(\omega)=\left[\begin{array}{cccc}
1 & 0 & \frac{\sin\omega\iota}{\omega} & -\frac{1-\cos\omega\iota}{\omega}\\
0 & 1 & \frac{1-\cos\omega\iota}{\omega} & \frac{\sin\omega\iota}{\omega}\\
0 & 0 & \cos\omega\iota & -\sin\omega\iota\\
0 & 0 & \sin\omega\iota & \cos\omega\iota
\end{array}\right]\!\!,\:\: & G=\left[\begin{array}{cc}
\frac{\iota^{2}}{2} & 0\\
0 & \frac{\iota^{2}}{2}\\
\iota & 0\\
0 & \iota
\end{array}\right]\!\!,
\end{alignat*}
$\iota=1s$, $w_{t}\sim\mathcal{N}(\cdot;0,\sigma_{w}^{2}I_{2})$,
$\sigma_{w}=0.5m/s^{2}$, $u_{t}\sim\mathcal{N}(\cdot;0,\sigma_{u}^{2})$,
$\sigma_{u}^ {}=(\pi/360)rad/s$. Object births follow a 4 component
LMB model $\{P_{B},p_{B}(\cdot,\ell_{i})\}_{i=1}^{4}$ where $P_{B}=0.01$
and $p_{B}(\cdot,\ell_{i})=\mathcal{N}(\cdot;m^{(i)},C)$ with $m^{(1)}=[20,80,0,0,0]^{T}$,
$m^{(2)}=[70,70,0,0,0]^{T}$, $m^{(3)}=[80,20,0,0,0]^{T}$, $m^{(4)}=[30,30,0,0,0]^{T}$,
$C=diag([1,1,1,1,\pi/180]^{2})$. 

The scenario duration is 100s with 4 trajectories appearing and disappearing
at different times as illustrated in Figure \ref{fig:gt}. There is
one birth from each of the birth locations respectively; label (1,
1) is present between $t=1:70s$, label (1, 2) is present between
$t=1:80s$, label (20, 3) is present between $t=20:100s$, label (30,
4) is present between $t=30:100s$. Thus 4 objects are simultaneously
present between $t=30:70s$. Figure \ref{fig:meas} shows samples
of the image observations for $I_{0}=15$. Note the crossing of 3
trajectories at time $t=60s$ results in nearby pixels having non-negligible
amplitude contributions from multiple objects. 

The multi-object posterior Gibbs sampler is run for 100 burn-in steps
and the subsequent 1000 samples are retained to represent the posterior
distribution. Successive iterations of the Gibbs sampler are run with
alternating forward and backward time passes which empirically is
seen to induce improved mixing. Initialization is carried out with
the factor sampling routine,\textcolor{blue}{{} }which empirically appears
to reduce burn-in time significantly, at least compared to an empty
or trivial initialization.\textcolor{violet}{{} }Sampling from unnormalized
single object posterior is implemented via particle flow or log homotopy
(refer to Appendix \ref{subsec:Particle-Flow} for details). 

Amongst the many multi-object posterior inference tasks, the most
popular and intuitive is multi-object trajectory estimation. For this
we adopt the \textit{label-MaM} estimate \cite{VoVoNguyenShimo-Overview24},
which is based on the label set $L^{*}$ with highest joint existence
probability over the duration of interest (if this label set is not
unique we select one with the most probable cardinality). Since each
$\ell$ $\in L^{\ast}$, has a range of possible trajectory supports
(the set of instances where this label exists), we use the most probable
support, and the mean of the attribute density at each instance of
this support.

\subsection{Multi-Object Trajectory Estimates}

For the purposes of providing a baseline performance validation, we
compare the multi-object trajectory estimate from the posterior with
superpositional measurement against the standard detection-based GLMB
smoother \cite{VoVomultiscan18}. To facilitate the comparison, a
simple fixed threshold detector is used, followed by basic clustering
on the detector output. The result is a set of noisy range and bearing
detections. The noise on the detections are modeled as additive zero
mean Gaussian with range and bearing standard deviations of $0.71m$
and $0.01rad$, respectively. Thus, for a fixed detection threshold
of $4$, this implies a clutter rate of $\lambda_{c}=0.32$. For source
intensity values ${\normalcolor {\normalcolor I_{0}=15,20,25,30,35}}$,
the corresponding detection probabilities (determined empirically
from the data) are ${\normalcolor {\normalcolor P_{D}=0.04,0.15,0.37,0.66,0.87}}$. 

Figures \ref{fig:estimate}, \ref{fig:est} and \ref{fig:estimate-1},
\ref{fig:est-1}, respectively, show typical estimated multi-object
trajectories from superpositional measurements and detection measurements,
for $I_{0}=15$. Observe that the posterior based on superpositional
measurements appears to facilitate more accurate trajectory estimates
in terms of initiation, following and termination. To further validate
these observations, 100 Monte Carlo trials are performed. The OSPA/OSPA\textsuperscript{2}
metrics are employed as a formal distance to assess the error between
the true and estimated multi-object states/trajectories \cite{Schumacher2008},
\cite{Beard18-largescale}. Intuitively, the OSPA\textsuperscript{2}
jointly captures state, cardinality and label errors, and is typically
plotted over a moving time window, in this case of length $10s$.
The OSPA metric is a special case with a window length of 1 which
effectively ignores labeling errors. 
\begin{figure}[t]
\begin{centering}
\resizebox{80mm}{!}{\includegraphics[clip]{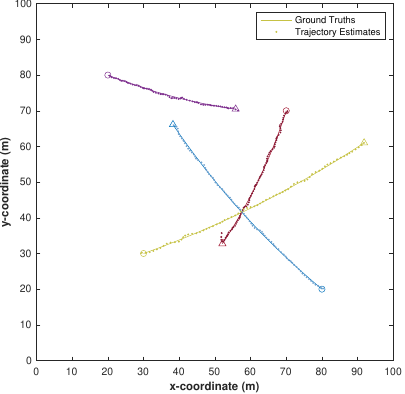}}
\par\end{centering}
\caption{\protect\label{fig:estimate}Multi-object trajectory estimate from
posterior with superpositional measurements, superimposed on ground
truth.}
\vspace{0mm}
\end{figure}
\begin{figure}[t]
\begin{centering}
\resizebox{88mm}{!}{\includegraphics[clip]{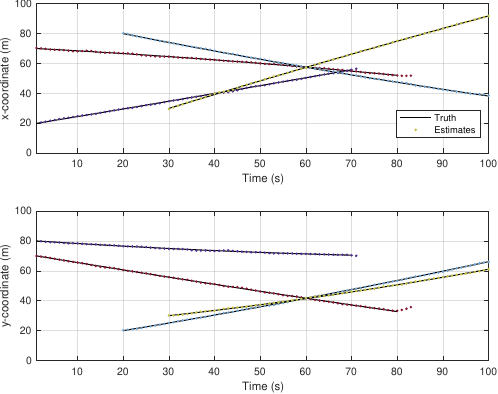}}
\par\end{centering}
\caption{\protect\label{fig:est}Multi-object trajectory estimate from superpositional-measurement-based
posterior versus time.}
\vspace{0mm}
\end{figure}
\begin{figure}[t]
\begin{centering}
\resizebox{80mm}{!}{\includegraphics[clip]{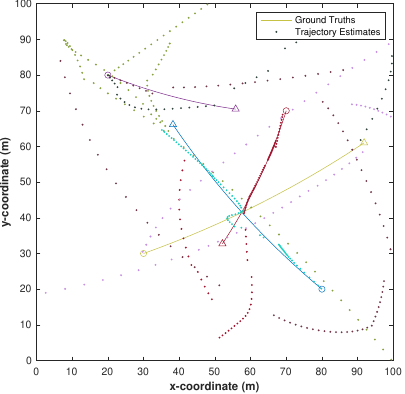}}
\par\end{centering}
\caption{\protect\label{fig:estimate-1}Multi-object trajectory estimate from
GLMB smoother with a simple detector, superimposed on ground truth\textcolor{blue}{.}}
\vspace{0mm}
\end{figure}
\begin{figure}[t]
\begin{centering}
\resizebox{88mm}{!}{\includegraphics[clip]{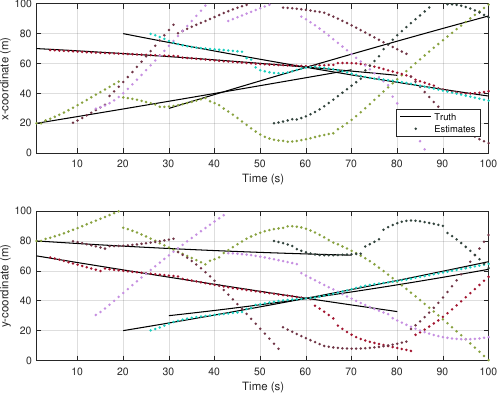}}
\par\end{centering}
\caption{\protect\label{fig:est-1}Estimates from GLMB smoother with a simple
detector versus time.}
\vspace{0mm}
\end{figure}

Figures \ref{fig:ospas_tbd} and \ref{fig:ospa_tad} show the average
OSPA/OSPA\textsuperscript{2} error for the multi-object posterior
Gibbs sampler and detection based GLMB smoother respectively. It can
be readily seen that the image based Gibbs sampler matches or outperforms
the detection based GLMB smoother in all cases presented. Notice that
it is only in the highest SNR setting where the detection based GLMB
smoother produces accurate trajectory estimates. The detection based
GLMB smoother succumbs to label switching earlier, and similarly fails
to produce meaningful trajectory estimation earlier. At lower SNR,
the superpositional measurement based Gibbs sampler significantly
outperforms the detection based GLMB smoother, also as expected. It
is also noted that at lower SNR, the former still produces relatively
accurate position estimates, albeit with more label switching, and
late track initiation/termination. 
\begin{figure}[t]
\begin{centering}
\resizebox{88mm}{!}{\includegraphics[clip]{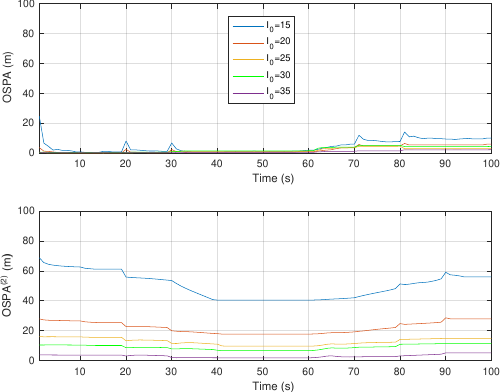}}
\par\end{centering}
\caption{\protect\label{fig:ospas_tbd}Multi-Object posterior Gibbs sampler:
OSPA and OSPA$^{(2)}$ tracking errors (with Euclidean base distance
and cut-off at 100m) versus time for various source levels. The tracking
errors increase as the source levels decrease.}
\vspace{0mm}
\end{figure}
\begin{figure}[t]
\begin{centering}
\resizebox{88mm}{!}{\includegraphics[clip]{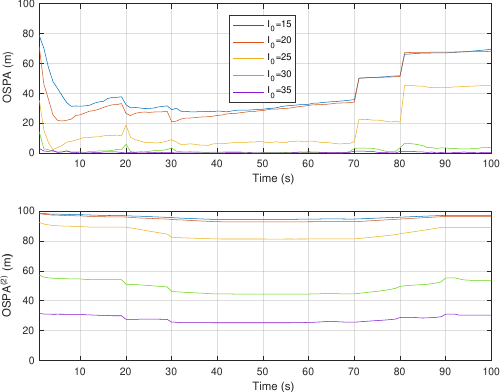}}
\par\end{centering}
\caption{\protect\label{fig:ospa_tad}GLMB smoother with a simple detector:
OSPA and OSPA$^{(2)}$ tracking errors (with Euclidean base distance
and cut-off at 100m) versus time for various source levels. Again,
the tracking errors increase as the source levels decrease. Further,
the tracking errors are significantly higher than that of the multi-object
posterior Gibbs sampler.}
\vspace{0mm}
\end{figure}

These observations can be directly attributed to the Gibbs sampler
fully exploiting the amplitude variations in the image measurement.
The GLMB smoother by nature only considers point detections which
discard significant information in the image. The respective approaches
effectively represent the two ends of the computation versus performance
trade off. At an algorithmic level, the image based Gibbs sampler
breaks a complex multi-scan multi-object image estimation problem
into a sequence of simpler single-scan single-object image estimation
problems. In contrast, the detection based GLMB smoother uses a simpler
and cheaper multi-object data association solver.

\subsection{Posterior Analytics}

Thanks to the labeled representation, beyond trajectory estimation,
the multi-object posterior admits rich analytics inaccessible to point-estimate-based
or unlabeled methods. The multi-object posterior samples provide a
direct and efficient means to compute these analytics. This subsection
presents some illustrative examples, motivated by tasks associated
with situational awareness, which are by no means an exhaustive exploration
of possible analytics.

For a given logical statement or event $\mathcal{E}(\boldsymbol{X}_{0:k})$
concerning a multi-object trajectory $\boldsymbol{X}_{0:k}$, we use
the notation $\boldsymbol{1}(\mathcal{E}(\boldsymbol{X}_{0:k}))=1$
if $\mathcal{E}(\boldsymbol{X}_{0:k})$ is true and zero otherwise.
The posterior probability of the event $\mathcal{E}$ given by
\begin{alignat*}{1}
P(\mathcal{E}) & =\mathbb{E}_{\pi_{0:k}}\left[\boldsymbol{1}(\mathcal{E})\right].
\end{alignat*}
Recall the label projection $\mathcal{L}(\cdot)$ and attribute projection
$\mathcal{A}(\cdot)$, we now denote by $\mathcal{A}_{p}$ and $\mathcal{A}_{v}$
respectively,\textcolor{violet}{{} }the attribute projections pertaining
to position and velocity. 

The first example considers a given time $t$, label set $L$, radius
$\rho$, and the event $\mathcal{E}_{t}(\boldsymbol{X}_{0:k})$, where
at time $t$, objects with labels in $L$ are all located within $\rho$
of each other. Let $d_{p}(\boldsymbol{x},\boldsymbol{y})=\left\Vert \mathcal{A}_{p}(\boldsymbol{x})-\mathcal{A}_{p}(\boldsymbol{y})\right\Vert $
denote the distance between the positions, then 
\begin{alignat*}{1}
\boldsymbol{1}(\mathcal{E}_{t}(\boldsymbol{X}_{0:k}))= & \boldsymbol{1}_{\mathcal{L}(\boldsymbol{X}_{t})}^{L}\boldsymbol{1}\!\left(\max_{\boldsymbol{x}\neq\boldsymbol{y}\in\boldsymbol{X}_{t}}\boldsymbol{1}_{L}^{\mathcal{L}(\{\boldsymbol{x},\boldsymbol{y}\})}d_{p}(\boldsymbol{x},\boldsymbol{y})\leq\rho\right)\!.
\end{alignat*}
For $L=\{(1,2),(20,3),(30,4)\}$, and $\rho=3m$, the posterior probability
of $\mathcal{E}_{t}(\boldsymbol{X}_{0:k})$ is plotted against $t$
on the interval ${\normalcolor {\normalcolor {\normalcolor \{1:k}\}}}$
in Figure \ref{fig:probs1}. Observed that this probability peaks
around $60s$ (and is otherwise low), which is consistent with the
crossing of the 3 trajectories in the ground truth. 

\begin{figure}[t]
\begin{centering}
\resizebox{88mm}{!}{\includegraphics[clip]{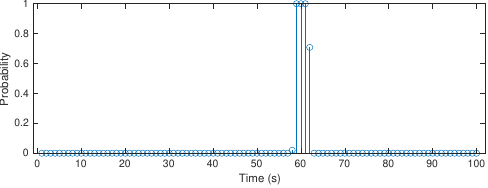}}
\par\end{centering}
\caption{\protect\label{fig:probs1}Posterior probability that the objects
(with labels) $(1,2)$, $(20,3)$, and\textcolor{blue}{{} }$(30,4)$
are all within $3m$ of each other.}
\vspace{0mm}
\end{figure}

The second example considers a given a time $t$, label set $L$,
tolerance $\eta$, and the event $\mathcal{E}_{t}(\boldsymbol{X}_{0:k})$,
where at time $t$, all objects with labels in $L$ have speeds within
$\eta$ of each other. Let $d_{s}(\boldsymbol{x},\boldsymbol{y})=\left|\left\Vert \mathcal{A}_{v}(\boldsymbol{x})\right\Vert -\left\Vert \mathcal{A}_{v}(\boldsymbol{y})\right\Vert \right|$
denote the distance between the speeds, then
\begin{alignat*}{1}
{\normalcolor \boldsymbol{1}(\mathcal{E}_{t}(\boldsymbol{X}_{0:k})}\mathclose{\normalcolor )} & \mathrel{\normalcolor =}{\normalcolor \boldsymbol{1}_{\mathcal{L}(\boldsymbol{X}_{t})}^{L}\boldsymbol{1}\!\left(\max_{\boldsymbol{x}\neq\boldsymbol{y}\in\boldsymbol{X}_{t}}\boldsymbol{1}_{L}^{\mathcal{L}(\{\boldsymbol{x},\boldsymbol{y}\})}d_{s}(\boldsymbol{x},\boldsymbol{y})\leq\eta\right)\!.}
\end{alignat*}
For $L=\{(20,3),(30,4)\}$, and $\eta=0.1m/s$, the posterior probability
of $\mathcal{E}_{t}(\boldsymbol{X}_{0:k})$ is plotted against $t$
on the interval ${\normalcolor {\normalcolor \{1:k\}}}$ in Figure
\ref{fig:probs2}. This probability is elevated after 30s when both
labels are present, but fluctuates due to the large uncertainty on
the speed estimates. 
\begin{figure}[t]
\begin{centering}
\resizebox{88mm}{!}{\includegraphics[clip]{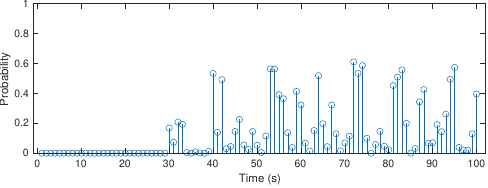}}
\par\end{centering}
\caption{\protect\label{fig:probs2} Posterior probability that objects $(20,3)$
and $(30,4)$ travel at speeds within $0.1m/s$ of each other.}
\vspace{0mm}
\end{figure}

The third example considers a given a time window $T(t)\subseteq\{0:k\}$,
label set $L$, radius $\rho$, angle $\theta$, and the event $\mathcal{E}_{t}(\boldsymbol{X}_{0:k})$,
where for all time in $T(t)$ all objects with labels in $L$ are
located more than $\rho$ from each other, and have differences in
headings greater than $\theta$. Let $d_{h}(\boldsymbol{x},\boldsymbol{y})=\arccos\left(\frac{\left\langle \mathcal{A}_{v}(\boldsymbol{x}),\mathcal{A}_{v}(\boldsymbol{y})\right\rangle }{\left\Vert \mathcal{A}_{v}(\boldsymbol{x})\right\Vert \left\Vert \mathcal{A}_{v}(\boldsymbol{y})\right\Vert }\right)$
denote the distance between the headings, then 
\begin{alignat*}{1}
\boldsymbol{1}(\mathcal{E}_{t}(\boldsymbol{X}_{0:k})) & =\prod_{j\in T(t)}\!\left[\boldsymbol{1}\!\left(\max_{\boldsymbol{x}\neq\boldsymbol{y}\in\boldsymbol{X}_{j}}\boldsymbol{1}_{L}^{\mathcal{L}(\{\boldsymbol{x},\boldsymbol{y}\})}d_{p}(\boldsymbol{x},\boldsymbol{y})>\rho\right)\right.\\
 & \left.\times\boldsymbol{1}\!\left(\max_{\boldsymbol{x}\neq\boldsymbol{y}\in\boldsymbol{X}_{j}}\!\boldsymbol{1}_{L}^{\mathcal{L}(\{\boldsymbol{x},\boldsymbol{y}\})}d_{h}(\boldsymbol{x},\boldsymbol{y})>\theta\right)\boldsymbol{1}_{\mathcal{L}(\boldsymbol{X}_{t})}^{L}\!\right]\!.
\end{alignat*}
For $T(t)=\{t-2,t-1,t\}$, $L=\{(1,2),(30,4)\}$, $\rho=10m$ and
$\theta=60^{\circ}$, the posterior probability of $\mathcal{E}_{t}(\boldsymbol{X}_{0:k})$
is plotted\textcolor{violet}{{} }against $t$ on the interval ${\normalcolor {\normalcolor {\normalcolor \{3:k\}}}}$
in Figure \ref{fig:probs3}. This probability is zero before $33s$
and after $80s$ since both labels do not exist concurrently for the
3 preceding scans. The probability dips between $50s$ and $70s$
due to the trajectory crossing, but is otherwise high as expected.

\begin{figure}[t]
\begin{centering}
\resizebox{88mm}{!}{\includegraphics[clip]{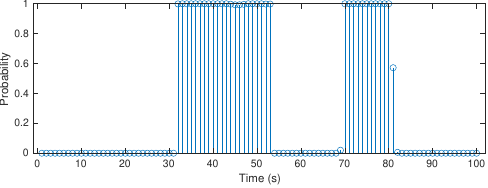}}
\par\end{centering}
\caption{\protect\label{fig:probs3}Posterior probability that for each time
in the window $T(t)=\{t-2,t-1,t\}$, objects $(1,2)$ and $(30,4)$
are more than $10m$ from each other, and have heading difference
of more than $60^{\circ}$.}
\vspace{0mm}
\end{figure}

\section{Conclusions\protect\label{sec:Conclusion}}

This work provides a tractable Gibbs sampling solution for direct
multi-object posterior computation under generic observation models,
including cases where no analytic filtering or smoothing solutions
exist. The proposed approach advances the state-of-the-art beyond
analytic GLMB-based methods, enabling principled Bayesian inference
directly at the posterior level. A key implication of this development
is the first multi-object smoothing solution for superpositional measurements,
potentially closing a long-standing gap between filtering-based TBD
methods and full Bayesian smoothing. From a practical perspective,
the proposed approach enables robust multi-object estimation in low-SNR
regimes where detection-based methods fail. Beyond trajectory estimates,
since the labeled RFS formulation admits rich analytics with intrinsic
uncertainty quantification, the posterior samples provide a direct
and efficient means to realize these analytics.

\section{Appendix\protect\label{sec:Appendix}}

\subsection{Mathematical Proofs\protect\label{sec:Appendix-proof}}

The following Lemmas facilitate the proofs.
\begin{lem}
\label{Lemma:LMB-Constraint-ID} Let $r$ be a function from a discrete
set $\mathbb{S}$ to $[0,1]$, and\textcolor{violet}{{} ${\normalcolor r_{I}(\ell)=r(\ell}\mathclose{\normalcolor )^{1-\mathbf{1}_{I}(\ell)}}$},
 Then for any $I,L\subseteq\mathbb{S}$
\begin{alignat*}{1}
\frac{\mathbf{1}_{L}^{I}\left[1-r\right]^{\mathbb{S}-L}r^{L}}{\sum_{J\subseteq\mathbb{S}}\mathbf{1}_{J}^{I}\left[1-r\right]^{\mathbb{S}-J}r^{J}}= & \mathbf{1}_{L}^{I}\left[1-r_{I}\right]^{\mathbb{S}-L}r_{I}^{L}.
\end{alignat*}
\end{lem}
\begin{IEEEproof}
Denote $\bar{J}=\mathbb{S}-J$. If $I\subseteq J\subseteq\mathbb{S}$,
then $J\cap I=I$, i.e., $r^{J}=r^{J\cap I}r^{J-I}=r^{J-I}r^{I}$.
Moreover, noting that $\mathbf{1}_{J}^{I}=1$ iff $J\supseteq I$,
we have
\begin{alignat*}{1}
\sum_{J\subseteq\mathbb{S}}\mathbf{1}_{J}^{I}\left[1-r\right]^{\bar{J}}r^{J}= & \sum_{J\supseteq I}\left[1-r\right]^{\bar{J}}r^{J-I}r^{I}\\
= & r^{I}\sum_{U\subseteq\mathbb{S}-I}\left[1-r\right]^{(\mathbb{S}-I)-U}r^{U}\\
= & r^{I}
\end{alignat*}
where: the 2nd line follows by the change of variable $U=J-I$ (i.e.,
the sum over $J\supseteq I$ translates to the sum over $U\subseteq\mathbb{S}-I$,
and $\bar{J}=(\mathbb{S}-I)-U$ because $J=U\uplus I$); and the 3rd
line follows from the multinomial expansion identity $\sum_{U\subseteq S}\left[1-r\right]^{S-U}r^{U}=1$,
see Lemma 1 in \cite{VoVoNguyenShimo-Overview24}.

Further, since the result trivially holds for $I\nsubseteq L$, assume
$I\subseteq L\subseteq\mathbb{S}$. Then $\bar{L}-I=\bar{L}$, and
$\bar{L}\subseteq\bar{I}$. Hence, $r_{I}^{L}=r^{L-I}$, and $\left[1-r_{I}\right]^{\bar{L}}=\left[1-r\right]^{\bar{L}}$,
because $r_{I}(\ell)=r(\ell)$ when $\ell\notin I$ and $r_{I}(\ell)=1$
when $\ell\in I$ . Therefore, 
\begin{alignat*}{1}
\frac{\mathbf{1}_{L}^{I}\left[1-r\right]^{\bar{L}}r^{L}}{\sum_{J\subseteq\mathbb{S}}\mathbf{1}_{J}^{I}\left[1-r\right]^{\bar{J}}r^{J}}= & \frac{\mathbf{1}_{L}^{I}\left[1-r\right]^{\bar{L}}r^{L}}{r^{I}}\\
= & \mathbf{1}_{L}^{I}\left[1-r\right]^{\bar{L}}r^{L-I}\\
= & \mathbf{1}_{L}^{I}\left[1-r_{I}\right]^{\bar{L}}r_{I}^{L}.
\end{alignat*}
\end{IEEEproof}
\begin{lem}
\label{Lemma:LMB-ID-}Let $f$ and $g$ be non-negative functions
on $\mathbb{X}\times\mathbb{S}$ and $\mathbb{S}$, respectively.
Denote $\langle f\rangle(\ell)=\langle f(\cdot,\ell),1\rangle$ and
assuming there is no $\ell\in\mathbb{S}$ such that $\langle f\rangle(\ell)=g(\ell)=0$.
Then 
\begin{alignat*}{1}
\Delta(\boldsymbol{X})\mathbf{1}_{\mathcal{\mathcal{L}}(\boldsymbol{X})}^{I}g^{\mathbb{S}-\mathcal{L}(\boldsymbol{X})}f^{\boldsymbol{X}} & \propto\mathbf{1}_{\mathcal{\mathcal{L}}(\boldsymbol{X})}^{I}\boldsymbol{\lambda}(\boldsymbol{X};r_{I},f_{I}),
\end{alignat*}
where $\boldsymbol{\lambda}(\cdot;r_{I},f_{I})$ is an LMB with parameters:
\begin{align*}
r_{I}(\ell)= & \left\lceil \frac{\langle f\rangle(\ell)}{g(\ell)}\right\rfloor ^{1-\mathbf{1}_{I}(\ell)}\\
f_{I}(\cdot,\ell)= & \frac{f(\cdot,\ell)}{\langle f\rangle(\ell)}.
\end{align*}
\end{lem}
\begin{IEEEproof}
Let $h_{I}(\boldsymbol{X})=\Delta(\boldsymbol{X})\mathbf{1}_{\mathcal{\mathcal{L}}(\boldsymbol{X})}^{I}g^{\mathbb{S}-\mathcal{L}(\boldsymbol{X})}f^{\boldsymbol{X}}$.
The result holds trivially when $\boldsymbol{X}$ does not have distinct
labels. Assume $\boldsymbol{X}$ has distinct labels, then 
\begin{align*}
h_{I}(\boldsymbol{X})= & \mathbf{1}_{\mathcal{\mathcal{L}}(\boldsymbol{X})}^{I}g^{\mathbb{S}-\mathcal{L}(\boldsymbol{X})}f^{\boldsymbol{X}}\\
= & \mathbf{1}_{\mathcal{\mathcal{L}}(\boldsymbol{X})}^{I}g^{\mathbb{S}-\mathcal{L}(\boldsymbol{X})}\langle f\rangle^{\mathcal{L}(\boldsymbol{X})}f_{I}^{\boldsymbol{X}}\\
= & \mathbf{1}_{\mathcal{\mathcal{L}}(\boldsymbol{X})}^{I}g^{\mathbb{S}}\left[\frac{\langle f\rangle}{g}\right]^{\mathcal{L}(\boldsymbol{X})}f_{I}^{\boldsymbol{X}}.
\end{align*}
Let $r=\left\lceil \frac{\langle f\rangle}{g}\right\rfloor =\frac{\langle f\rangle}{\langle f\rangle+g}$,
then $\langle f\rangle=rg+r\langle f\rangle$, and hence $\frac{\langle f\rangle}{g}=\frac{r}{1-r}$
. Substituting this into the above equation gives
\begin{align*}
h_{I}(\boldsymbol{X})= & \mathbf{1}_{\mathcal{\mathcal{L}}(\boldsymbol{X})}^{I}g^{\mathbb{S}}\left[\frac{r}{1-r}\right]^{\mathcal{L}(\boldsymbol{X})}f_{I}^{\boldsymbol{X}}\\
= & \mathbf{1}_{\mathcal{\mathcal{L}}(\boldsymbol{X})}^{I}\left[\frac{g}{1-r}\right]^{\mathbb{S}}\left[1-r\right]^{\mathbb{S}-\mathcal{L}(\boldsymbol{X})}r^{\mathcal{L}(\boldsymbol{X})}f_{I}^{\boldsymbol{X}}\\
\propto & \left[\frac{g}{1-r}\right]^{\mathbb{S}}\mathbf{1}_{\mathcal{\mathcal{L}}(\boldsymbol{X})}^{I}\left[1-r_{I}\right]^{\mathbb{S}-\mathcal{L}(\boldsymbol{X})}r_{I}^{\mathcal{L}(\boldsymbol{X})}f_{I}^{\boldsymbol{X}},
\end{align*}
where the last line follows from Lemma \ref{Lemma:LMB-Constraint-ID}.
Lemma \ref{Lemma:LMB-Constraint-ID} also implies $\sum_{L\subseteq\mathbb{S}}\mathbf{1}_{L}^{I}\left[1-r_{I}\right]^{\mathbb{S}-L}r_{I}^{L}=1$.
Hence, for $\mathcal{\mathcal{L}}(\boldsymbol{X})\supseteq I$, $\left[1-r_{I}\right]^{\mathbb{S}-\mathcal{L}(\boldsymbol{X})}r_{I}^{\mathcal{L}(\boldsymbol{X})}f_{I}^{\boldsymbol{X}}$
is an LMB, and the desired result holds.
\end{IEEEproof}

\subsubsection{Proof of Proposition \ref{Prop:Transition-Prod-LMB}}

Using (\ref{eq:multiobject-trans}), and keeping in mind that we are
interested in expressing the product as a function of $\boldsymbol{X}$,
\begin{alignat*}{1}
\boldsymbol{f}_{\texttt{+}}(\boldsymbol{X}_{\texttt{+}}|\boldsymbol{X})\boldsymbol{f}(\boldsymbol{X}|\boldsymbol{X}_{\texttt{-}})\\
 & \negthickspace\negthickspace\negthickspace\negthickspace\negthickspace\negthickspace\negthickspace\negthickspace\negthickspace\negthickspace\negthickspace\negthickspace\negthickspace\negthickspace\negthickspace\negthickspace=\boldsymbol{f}_{B,\texttt{+}}(\boldsymbol{X}_{\texttt{+}}^{\mathbb{\mathbb{B}}_{\texttt{+}}})\boldsymbol{f}_{S,\texttt{+}}(\boldsymbol{X}_{\texttt{+}}^{\mathbb{\mathbb{L}}}|\boldsymbol{X})\boldsymbol{f}(\boldsymbol{X}|\boldsymbol{X}_{\texttt{-}})\\
 & \negthickspace\negthickspace\negthickspace\negthickspace\negthickspace\negthickspace\negthickspace\negthickspace\negthickspace\negthickspace\negthickspace\negthickspace\negthickspace\negthickspace\negthickspace\negthickspace\propto\boldsymbol{f}_{S,\texttt{+}}(\boldsymbol{X}_{\texttt{+}}^{\mathbb{\mathbb{L}}}|\boldsymbol{X})\boldsymbol{f}(\boldsymbol{X}|\boldsymbol{X}_{\texttt{-}})\\
 & \negthickspace\negthickspace\negthickspace\negthickspace\negthickspace\negthickspace\negthickspace\negthickspace\negthickspace\negthickspace\negthickspace\negthickspace\negthickspace\negthickspace\negthickspace\negthickspace=\boldsymbol{f}_{S,\texttt{+}}(\boldsymbol{X}_{\texttt{+}}^{\mathbb{\mathbb{L}}}|\boldsymbol{X}^{\mathbb{\mathbb{B}}}\uplus\boldsymbol{X}^{\mathbb{\mathbb{L_{\texttt{-}}}}})\boldsymbol{f}(\boldsymbol{X}^{\mathbb{\mathbb{B}}}\uplus\boldsymbol{X}^{\mathbb{\mathbb{L_{\texttt{-}}}}}|\boldsymbol{X}_{\texttt{-}})\\
 & \negthickspace\negthickspace\negthickspace\negthickspace\negthickspace\negthickspace\negthickspace\negthickspace\negthickspace\negthickspace\negthickspace\negthickspace\negthickspace\negthickspace\negthickspace\negthickspace=\boldsymbol{f}_{S,\texttt{+}}(\boldsymbol{X}_{\texttt{+}}^{\mathbb{\mathbb{L}}}|\boldsymbol{X}^{\mathbb{\mathbb{B}}}\uplus\boldsymbol{X}^{\mathbb{\mathbb{L_{\texttt{-}}}}})\boldsymbol{f}_{B}(\boldsymbol{X}^{\mathbb{\mathbb{B}}})\boldsymbol{f}_{S}(\boldsymbol{X}^{\mathbb{\mathbb{L_{\texttt{-}}}}}|\boldsymbol{X}_{\texttt{-}}).
\end{alignat*}
For readability, we assume $\boldsymbol{X},\boldsymbol{X}_{\texttt{-}},\boldsymbol{X}_{\texttt{+}}^{\mathbb{\mathbb{L}}}$
all have distinct labels, and omit the distinct label indicator. Recall
$\mathcal{L}(\boldsymbol{X}^{\mathbb{\mathbb{B}}})\subseteq\mathcal{L}(\mathbb{B})$,
$\mathcal{L}(\boldsymbol{X}^{\mathbb{\mathbb{L_{\texttt{-}}}}})\subseteq\mathcal{L}(\boldsymbol{X}_{\texttt{-}})$,
and (\ref{eq:birthRFS1})-(\ref{eq:survival-RFS-X-2}), we write
\begin{alignat}{1}
\boldsymbol{f}_{B}(\boldsymbol{X}^{\mathbb{\mathbb{B}}}) & ={\color{red}\mathbf{1}_{\mathbb{B}}^{\mathcal{L}(\boldsymbol{X}^{\mathbb{\mathbb{B}}})}Q_{B}^{\mathbb{B}-\mathcal{L}(\boldsymbol{X}^{\mathbb{\mathbb{B}}})}b^{\boldsymbol{X}^{\mathbb{\mathbb{B}}}}}\nonumber \\
 & ={\color{red}Q_{B}^{\mathbb{B}-\mathcal{L}(\boldsymbol{X}^{\mathbb{\mathbb{B}}})}b^{\boldsymbol{X}^{\mathbb{\mathbb{B}}}}}\label{eq:density-of-birth}\\
\boldsymbol{f}_{S}(\boldsymbol{X}^{\mathbb{\mathbb{L_{\texttt{-}}}}}|\boldsymbol{X}_{\texttt{-}}) & ={\color{blue}\mathbf{1}_{\mathcal{L}(\boldsymbol{X}_{\texttt{-}})}^{\mathcal{L}(\boldsymbol{X}^{\mathbb{\mathbb{L_{\texttt{-}}}}})}Q_{S}^{\boldsymbol{X}_{\texttt{-}}^{\mathcal{L}\left(\boldsymbol{X}_{\texttt{-}}\right)-\mathcal{L}\left(\boldsymbol{X}\right)}}\overset{\shortleftarrow}{f}(\cdot;\boldsymbol{X}_{\texttt{-}}}\mathclose{\color{blue})^{\boldsymbol{X}^{\mathbb{\mathbb{L_{\texttt{-}}}}}}}\nonumber \\
 & ={\color{blue}Q_{S}^{\boldsymbol{X}_{\texttt{-}}^{\mathcal{L}\left(\boldsymbol{X}_{\texttt{-}}\right)-\mathcal{L}\left(\boldsymbol{X}\right)}}\overset{\shortleftarrow}{f}(\cdot;\boldsymbol{X}_{\texttt{-}}}\mathclose{\color{blue})^{\boldsymbol{X}^{\mathbb{\mathbb{L_{\texttt{-}}}}}}}\label{eq:density-of-survival}
\end{alignat}
Further, using (\ref{eq:survival-RFS-Xminus-1}), (\ref{eq:survival-RFS-Xminus-2}),
and\textcolor{blue}{{} }noting that $\mathcal{L}(\boldsymbol{X}_{\texttt{+}}^{\mathbb{\mathbb{L}}})\subseteq\mathcal{\mathcal{L}}(\boldsymbol{X}^{\mathbb{\mathbb{B}}}\uplus\boldsymbol{X}^{\mathbb{\mathbb{L_{\texttt{-}}}}})$
iff $\mathcal{L}(\boldsymbol{X}_{\texttt{+}}^{\mathbb{\mathbb{L}}})\cap\mathbb{B}\subseteq\mathcal{\mathcal{L}}(\boldsymbol{X}^{\mathbb{\mathbb{B}}})$
and $\mathcal{L}(\boldsymbol{X}_{\texttt{+}}^{\mathbb{\mathbb{L}}})\cap\mathbb{\mathbb{L_{\texttt{-}}}}\subseteq\mathcal{\mathcal{L}}(\boldsymbol{X}^{\mathbb{\mathbb{L_{\texttt{-}}}}})$,
we write 
\begin{alignat}{1}
{\normalcolor {\normalcolor {\normalcolor {\normalcolor \boldsymbol{f}_{S,\texttt{+}}(\boldsymbol{X}_{\texttt{+}}^{\mathbb{\mathbb{L}}}|\boldsymbol{X}^{\mathbb{\mathbb{B}}}\uplus\boldsymbol{X}^{\mathbb{\mathbb{L_{\texttt{-}}}}}}\mathclose{\normalcolor )}}}} & {\normalcolor {\normalcolor {\normalcolor {\color{black}}}}}\nonumber \\
{\normalcolor {\normalcolor }} & {\normalcolor {\normalcolor \negthickspace\negthickspace\negthickspace\negthickspace\negthickspace\negthickspace\negthickspace\negthickspace\negthickspace\negthickspace\negthickspace\negthickspace\negthickspace\negthickspace\negthickspace\negthickspace\negthickspace\negthickspace\negthickspace\negthickspace\negthickspace\negthickspace\negthickspace\negthickspace{\normalcolor \mathrel{\normalcolor =}{\normalcolor \mathbf{1}_{\mathcal{\mathcal{L}}(\boldsymbol{X}^{\mathbb{\mathbb{B}}}\uplus\boldsymbol{X}^{\mathbb{\mathbb{L_{\texttt{-}}}}})}^{\mathcal{L}(\boldsymbol{X}_{\texttt{+}}^{\mathbb{\mathbb{L}}})}\overset{\shortrightarrow}{f}(\boldsymbol{X}_{\texttt{+}}^{\mathbb{\mathbb{L}}};\cdot}\mathclose{\normalcolor )^{\boldsymbol{X}^{\mathbb{\mathbb{B}}}\uplus\boldsymbol{X}^{\mathbb{\mathbb{L_{\texttt{-}}}}}}}}}}\nonumber \\
 & \negthickspace\negthickspace\negthickspace\negthickspace\negthickspace\negthickspace\negthickspace\negthickspace\negthickspace\negthickspace\negthickspace\negthickspace\negthickspace\negthickspace\negthickspace\negthickspace\negthickspace\negthickspace\negthickspace\negthickspace\negthickspace\negthickspace\negthickspace\negthickspace={\color{black}{\color{purple}{\color{blue}{\color{red}\mathbf{1}_{\mathcal{\mathcal{L}}(\boldsymbol{X}^{\mathbb{\mathbb{B}}})}^{\mathcal{L}(\boldsymbol{X}_{\texttt{+}}^{\mathbb{\mathbb{L}}})\cap\mathbb{B}}}{\color{red}{\color{red}{\color{blue}\mathbf{1}_{\mathcal{\mathcal{L}}(\boldsymbol{X}^{\mathbb{\mathbb{L_{\texttt{-}}}}})}^{\mathcal{L}(\boldsymbol{X}_{\texttt{+}}^{\mathbb{\mathbb{L}}})\cap\mathbb{\mathbb{L_{\texttt{-}}}}}}}}}}{\color{red}\overset{\shortrightarrow}{f}(\boldsymbol{X}_{\texttt{+}}^{\mathbb{\mathbb{L}}};\cdot}\mathclose{\color{red})^{\boldsymbol{X}^{\mathbb{\mathbb{B}}}}}}{\color{blue}\overset{\shortrightarrow}{f}(\boldsymbol{X}_{\texttt{+}}^{\mathbb{\mathbb{L}}};\cdot}\mathclose{\color{blue})^{\boldsymbol{X}^{\mathbb{\mathbb{L_{\texttt{-}}}}}}}\label{eq:density-of-next-survival}
\end{alignat}
Multiplying (\ref{eq:density-of-birth}), (\ref{eq:density-of-survival}),
(\ref{eq:density-of-next-survival}) and rearranging gives
\begin{alignat*}{1}
\boldsymbol{f}_{\texttt{+}}(\boldsymbol{X}_{\texttt{+}}|\boldsymbol{X})\boldsymbol{f}(\boldsymbol{X}|\boldsymbol{X}_{\texttt{-}})={\color{red}\boldsymbol{f}_{B}^{*}(\boldsymbol{X}^{\mathbb{\mathbb{B}}};\boldsymbol{X}_{\texttt{+}}^{\mathbb{\mathbb{L}}}}\mathclose{\color{red})}{\color{blue}\boldsymbol{f}_{S}^{*}(\boldsymbol{X}^{\mathbb{\mathbb{L_{\texttt{-}}}}};\boldsymbol{X}_{\texttt{+}}^{\mathbb{\mathbb{L}}},\boldsymbol{X}_{\texttt{-}}}\mathclose{\color{blue})}
\end{alignat*}
where\textcolor{black}{
\begin{alignat*}{1}
{\color{blue}{\color{red}\boldsymbol{f}_{B}^{*}(\boldsymbol{X}^{\mathbb{\mathbb{B}}};\boldsymbol{X}_{\texttt{+}}^{\mathbb{\mathbb{L}}})}\mathrel{\color{red}=}} & {\color{blue}{\color{red}{\color{black}{\color{red}{\color{red}\mathbf{1}_{\mathcal{\mathcal{L}}(\boldsymbol{X}^{\mathbb{\mathbb{B}}})}^{\mathcal{L}(\boldsymbol{X}_{\texttt{+}}^{\mathbb{\mathbb{L}}})\cap\mathbb{B}}}Q_{B}^{\mathbb{B}-\mathcal{L}(\boldsymbol{X}^{\mathbb{\mathbb{B}}})}\!}\mathinner{\color{red}\left[\overset{\shortrightarrow}{f}(\boldsymbol{X}_{\texttt{+}}^{\mathbb{\mathbb{L}}};\cdot)b(\cdot)\right]^{\boldsymbol{X}^{\mathbb{\mathbb{B}}}}}}}}\\
{\color{blue}{\color{blue}\boldsymbol{f}_{S}^{*}(\boldsymbol{X}^{\mathbb{\mathbb{L_{\texttt{-}}}}};\boldsymbol{X}_{\texttt{+}}^{\mathbb{\mathbb{L}}},\boldsymbol{X}_{\texttt{-}})}\mathrel{\color{blue}=}} & {\color{blue}{\color{red}{\color{blue}{\color{red}{\color{red}{\color{blue}\mathbf{1}_{\mathcal{\mathcal{L}}(\boldsymbol{X}^{\mathbb{\mathbb{L_{\texttt{-}}}}})}^{\mathcal{L}(\boldsymbol{X}_{\texttt{+}}^{\mathbb{\mathbb{L}}})\cap\mathbb{\mathbb{L_{\texttt{-}}}}}Q_{S}^{\boldsymbol{X}_{\texttt{-}}^{\mathcal{L}\left(\boldsymbol{X}_{\texttt{-}}\right)-\mathcal{L}\left(\boldsymbol{X}\right)}}}}}}}}\\
\mathbin{\color{blue}\times} & \mathinner{\color{blue}\left[\overset{\shortrightarrow}{f}(\boldsymbol{X}_{\texttt{+}}^{\mathbb{\mathbb{L}}};\cdot)\overset{\shortleftarrow}{f}(\cdot;\boldsymbol{X}_{\texttt{-}})\right]^{{\color{red}\!}\boldsymbol{X}^{\mathbb{\mathbb{L_{\texttt{-}}}}}}}
\end{alignat*}
}Using Lemma \ref{Lemma:LMB-ID-}\textcolor{black}{
\begin{alignat*}{1}
{\color{teal}{\normalcolor \boldsymbol{f}_{B}^{*}(\boldsymbol{X}^{\mathbb{\mathbb{B}}};\boldsymbol{X}_{\texttt{+}}^{\mathbb{\mathbb{L}}}}\mathclose{\normalcolor )}} & \mathrel{\normalcolor \propto}{\color{black}{\color{red}\mathbf{1}_{\mathcal{\mathcal{L}}(\boldsymbol{X}^{\mathbb{\mathbb{B}}})}^{\mathcal{L}(\boldsymbol{X}_{\texttt{+}}^{\mathbb{\mathbb{L}}})\cap\mathbb{B}}\boldsymbol{\lambda}(\boldsymbol{X}^{\mathbb{\mathbb{B}}};r_{\boldsymbol{X}_{\texttt{+}}},p_{\boldsymbol{X}_{\texttt{+}}}}\mathclose{\color{red})}}\\
{\color{teal}{\normalcolor \boldsymbol{f}_{S}^{*}(\boldsymbol{X}^{\mathbb{\mathbb{L_{\texttt{-}}}}};\boldsymbol{X}_{\texttt{+}}^{\mathbb{\mathbb{L}}},\boldsymbol{X}_{\texttt{-}}}\mathclose{\normalcolor )}} & \mathrel{\normalcolor \propto}{\color{red}{\color{red}{\color{blue}\mathbf{1}_{\mathcal{\mathcal{L}}(\boldsymbol{X}^{\mathbb{\mathbb{L_{\texttt{-}}}}})}^{\mathcal{L}(\boldsymbol{X}_{\texttt{+}}^{\mathbb{\mathbb{L}}})\cap\mathbb{\mathbb{L_{\texttt{-}}}}}\boldsymbol{\lambda}(\boldsymbol{X}^{\mathbb{\mathbb{L_{\texttt{-}}}}};r_{\boldsymbol{X}_{\texttt{+}},\boldsymbol{X}_{\texttt{-}}},p_{\boldsymbol{X}_{\texttt{+}},\boldsymbol{X}_{\texttt{-}}}).}}}
\end{alignat*}
}where
\begin{alignat*}{1}
{\normalcolor {\normalcolor {\normalcolor {\normalcolor {\normalcolor {\normalcolor {\normalcolor {\normalcolor {\normalcolor r_{\boldsymbol{X}_{\texttt{\textup{+}}}}(\ell}\mathclose{\normalcolor )}}}}}}}}} & {\normalcolor {\normalcolor {\normalcolor {\normalcolor {\normalcolor {\normalcolor {\normalcolor \mathrel{\normalcolor =}}}{\normalcolor \mathinner{\normalcolor \left\lceil \!\frac{\langle\overset{\shortrightarrow}{f}(\boldsymbol{X}_{\texttt{\textup{+}}}^{\mathbb{\mathbb{L}}};\cdot)b(\cdot)\rangle(\ell)}{Q_{B}(\ell)}\!\right\rfloor }^{{\normalcolor 1-\mathbf{1}_{\mathcal{L}(\boldsymbol{X}_{\texttt{\textup{+}}}^{\mathbb{\mathbb{L}}})\cap\mathbb{B}}(\ell}\mathclose{\normalcolor )}}}}\mathpunct{\normalcolor ,}}}}}\\
{\normalcolor {\normalcolor {\normalcolor {\normalcolor {\normalcolor {\normalcolor p_{\boldsymbol{X}_{\texttt{\textup{+}}}}(x,\ell}\mathclose{\normalcolor )}}}}}} & {\normalcolor {\normalcolor {\normalcolor {\normalcolor {\normalcolor \mathrel{\normalcolor =}}\frac{\overset{\shortrightarrow}{f}(\boldsymbol{X}_{\texttt{\textup{+}}}^{\mathbb{\mathbb{L}}};x,\ell)b(x,\ell)}{\langle\overset{\shortrightarrow}{f}(\boldsymbol{X}_{\texttt{\textup{+}}}^{\mathbb{\mathbb{L}}};\cdot)b(\cdot)\rangle(\ell)}}\mathpunct{\normalcolor ,}}}}\\
{\normalcolor {\normalcolor {\normalcolor {\normalcolor {\normalcolor {\normalcolor r_{\boldsymbol{X}_{\texttt{\textup{+}}},\boldsymbol{X}_{\texttt{-}}}(\ell}\mathclose{\normalcolor )}}}}}} & {\normalcolor {\normalcolor {\normalcolor {\normalcolor {\normalcolor \mathrel{\normalcolor =}}\left\lceil \!\frac{\langle\overset{\shortrightarrow}{f}(\boldsymbol{X}_{\texttt{\textup{+}}}^{\mathbb{\mathbb{L}}};\cdot)\overset{\shortleftarrow}{f}(\cdot;\boldsymbol{X}_{\texttt{-}})\rangle(\ell)}{Q_{S}^{\boldsymbol{X}_{\texttt{-}}^{\{\ell\}}}}\!\right\rfloor ^{1-\mathbf{1}_{\mathcal{L}(\boldsymbol{X}_{\texttt{\textup{+}}}^{\mathbb{\mathbb{L}}})\cap\mathbb{\mathbb{L_{\texttt{-}}}\!}}(\ell)}}\mathpunct{\normalcolor ,}}}}\\
{\normalcolor {\normalcolor {\normalcolor {\normalcolor {\normalcolor p_{\boldsymbol{X}_{\texttt{\textup{+}}},\boldsymbol{X}_{\texttt{-}}}(x,\ell}\mathclose{\normalcolor )}}}}} & {\normalcolor {\normalcolor {\normalcolor {\normalcolor \mathrel{\normalcolor =}\frac{\overset{\shortrightarrow}{f}(\boldsymbol{X}_{\texttt{\textup{+}}}^{\mathbb{\mathbb{L}}};x,\ell)\overset{\shortleftarrow}{f}(x,\ell;\boldsymbol{X}_{\texttt{-}})}{\langle\overset{\shortrightarrow}{f}(\boldsymbol{X}_{\texttt{\textup{+}}}^{\mathbb{\mathbb{L}}};\cdot)\overset{\shortleftarrow}{f}(\cdot;\boldsymbol{X}_{\texttt{-}})\rangle(\ell)}.}}}}
\end{alignat*}
Noting that for $\boldsymbol{\lambda}(\boldsymbol{X}^{\mathbb{\mathbb{B}}};r_{\boldsymbol{X}_{\texttt{+}}},p_{\boldsymbol{X}_{\texttt{+}}})$
we only need to evaluate $r_{\boldsymbol{X}_{\texttt{+}}}(\ell)$
at $\ell\in\mathcal{L}(\boldsymbol{X}^{\mathbb{\mathbb{B}}})$, where
$\ell\notin\mathbb{\mathbb{\mathbb{B}}}_{+}$ and hence $\mathbf{1}_{\mathcal{L}(\boldsymbol{X}_{\texttt{\textup{+}}}^{\mathbb{\mathbb{L}}})\cap\mathbb{B}}(\ell)=\mathbf{1}_{\mathcal{L}(\boldsymbol{X}_{\texttt{\textup{+}}}^{\mathbb{\mathbb{L}}})}(\ell)=|\boldsymbol{X}_{\texttt{\textup{+}}}^{\{\ell\}}|$.
Similarly, for $\boldsymbol{\lambda}(\boldsymbol{X}^{\mathbb{\mathbb{L_{\texttt{-}}}}};r_{\boldsymbol{X}_{\texttt{+}},\boldsymbol{X}_{\texttt{-}}},p_{\boldsymbol{X}_{\texttt{+}},\boldsymbol{X}_{\texttt{-}}})$
we only need to evaluate $r_{\boldsymbol{X}_{\texttt{\textup{+}}},\boldsymbol{X}_{\texttt{-}}}(\ell)$
at $\ell\in\mathcal{L}(\boldsymbol{X}^{\mathbb{\mathbb{L_{\texttt{-}}}}})$,
where $\ell\notin\mathbb{\mathbb{\mathbb{B}}}_{+}$ and hence $\mathbf{1}_{\mathcal{L}(\boldsymbol{X}_{\texttt{\textup{+}}}^{\mathbb{\mathbb{L}}})\cap\mathbb{\mathbb{L_{\texttt{-}}}\!}}(\ell)=\mathbf{1}_{\mathcal{L}(\boldsymbol{X}_{\texttt{\textup{+}}}^{\mathbb{\mathbb{L}}})}(\ell)=|\boldsymbol{X}_{\texttt{\textup{+}}}^{\{\ell\}}|$.
Further, since $\langle fg\rangle(\ell)=\langle f(\cdot,\ell),g(\cdot,\ell)\rangle$,
we have the desired result. $\blacksquare$

\subsubsection{Proof of Proposition \ref{Prop:Bernoulli-Conditional}}

Noting that $\boldsymbol{\lambda}(\boldsymbol{X};r,p)=\Delta(\boldsymbol{X})r^{\mathcal{L}(\boldsymbol{X})}\left[1-r\right]^{\mathbb{L}-\mathcal{L}(\boldsymbol{X})}p{}^{\boldsymbol{X}}$,
and let
\begin{alignat*}{1}
K_{\ell}(\boldsymbol{X}^{\{\bar{\ell}\}}) & =\Delta(\boldsymbol{X}^{\{\bar{\ell}\}})r^{\mathcal{L}(\boldsymbol{X}^{\{\bar{\ell}\}})}\left[1-r\right]^{\mathbb{L}-\mathcal{L}(\boldsymbol{X}^{\{\bar{\ell}\}})-\{\ell\}}p^{\boldsymbol{X}^{\{\bar{\ell}\}}},
\end{alignat*}
we have
\begin{alignat*}{1}
K_{\ell}(\boldsymbol{X}^{\{\bar{\ell}\}})(1-r(\ell)) & =\boldsymbol{\lambda}(\boldsymbol{X}^{\{\bar{\ell}\}};r,p)\\
K_{\ell}(\boldsymbol{X}^{\{\bar{\ell}\}})r(\ell)p(x,\ell) & =\boldsymbol{\lambda}(\{(x,\ell)\}\uplus\boldsymbol{X}^{\{\bar{\ell}\}};r,p).
\end{alignat*}
Further, let $C_{\ell}(\boldsymbol{X}^{\{\bar{\ell}\}})=\int\boldsymbol{\pi}(\boldsymbol{Z}^{\{\ell\}}\uplus\boldsymbol{X}^{\{\bar{\ell}\}})\delta\boldsymbol{Z}^{\{\ell\}},$
then
\begin{alignat*}{1}
C_{\ell}(\boldsymbol{X}^{\{\bar{\ell}\}})\\
 & \negthickspace\negthickspace\negthickspace\negthickspace\negthickspace\negthickspace\negthickspace\negthickspace\negthickspace\negthickspace\negthickspace\negthickspace\negthickspace\negthickspace\negthickspace\negthickspace=\boldsymbol{\pi}(\boldsymbol{X}^{\{\bar{\ell}\}})+\bigl\langle\boldsymbol{\pi}(\{\cdot\}\uplus\boldsymbol{X}^{\{\bar{\ell}\}})\bigr\rangle(\ell)\\
 & \negthickspace\negthickspace\negthickspace\negthickspace\negthickspace\negthickspace\negthickspace\negthickspace\negthickspace\negthickspace\negthickspace\negthickspace\negthickspace\negthickspace\negthickspace\negthickspace=K_{\ell}(\boldsymbol{X}^{\{\bar{\ell}\}\negthinspace})\negthinspace\negthinspace\left[\negthinspace(1\negthinspace-\negthinspace r(\ell))\phi(\boldsymbol{X}^{\{\bar{\ell}\}\negthinspace})\negthinspace+\negthinspace r(\ell)\bigl\langle p\phi(\{\negthinspace(\cdot)\negthinspace\}\negthinspace\uplus\negthinspace\boldsymbol{X}^{\{\bar{\ell}\}\negthinspace})\bigr\rangle(\ell)\negthinspace\right].
\end{alignat*}
Hence, 
\begin{alignat*}{1}
\boldsymbol{\pi}(\{(x,\ell)\}|\boldsymbol{X}^{\{\bar{\ell}\}})\\
 & \negthickspace\negthickspace\negthickspace\negthickspace\negthickspace\negthickspace\negthickspace\negthickspace\negthickspace\negthickspace\negthickspace\negthickspace\negthickspace\negthickspace\negthickspace\negthickspace=\boldsymbol{\pi}(\{(x,\ell)\}\uplus\boldsymbol{X}^{\{\bar{\ell}\}})/C_{\ell}(\boldsymbol{X}^{\{\bar{\ell}\}})\\
 & \negthickspace\negthickspace\negthickspace\negthickspace\negthickspace\negthickspace\negthickspace\negthickspace\negthickspace\negthickspace\negthickspace\negthickspace\negthickspace\negthickspace\negthickspace\negthickspace=\frac{\boldsymbol{\lambda}(\{(x,\ell)\}\uplus\boldsymbol{X}^{\{\bar{\ell}\}};r,p)\phi(\{(x,\ell)\}\uplus\boldsymbol{X}^{\{\bar{\ell}\}})}{C_{\ell}(\boldsymbol{X}^{\{\bar{\ell}\}})}\\
 & \negthickspace\negthickspace\negthickspace\negthickspace\negthickspace\negthickspace\negthickspace\negthickspace\negthickspace\negthickspace\negthickspace\negthickspace\negthickspace\negthickspace\negthickspace\negthickspace=\frac{r(\ell)p(\ell,x)\phi(\{(x,\ell)\}\uplus\boldsymbol{X}^{\{\bar{\ell}\}})}{(1-r(\ell))\phi(\boldsymbol{X}^{\{\bar{\ell}\}})+r(\ell)\bigl\langle p(\cdot)\phi(\{(\cdot)\}\uplus\boldsymbol{X}^{\{\bar{\ell}\}})\bigr\rangle(\ell)}.
\end{alignat*}
Since $r(\ell|\boldsymbol{X}^{\{\bar{\ell}\}})=\bigl\langle\boldsymbol{\pi}(\{(\cdot)\}|\boldsymbol{X}^{\{\bar{\ell}\}})\bigr\rangle(\ell)$,
and $p(x,\ell|\boldsymbol{X}^{\{\bar{\ell}\}})=\boldsymbol{\pi}(\{(x,\ell)\}|\boldsymbol{X}^{\{\bar{\ell}\}})/r(\ell|\boldsymbol{X}^{\{\bar{\ell}\}})$,
we have the desired result. $\blacksquare$

\subsection{Particle Flow\protect\label{subsec:Particle-Flow}}

Unlike Metropolis--Hastings, which requires long Markov chains and
proposal tuning \cite{Andrieu2003,RobertCasella2004}, or sequential
Monte Carlo, which may suffer from degeneracy in high dimensions \cite{Doucet01,Cappe07,DelMoral2006SMC,DoucetTutorial09},
Particle Flow when applicable, can avoid rejection steps, reduce variance,
and lower the number of required samples. A detailed overview is given
in \cite{DaiDaum2023StochasticPF}, while the original idea can be
traced back to the work by Daum in \cite{DaumHuang2007LogHomotopy}.

The key idea of Particle Flow is to directly transport samples along
a homotopy path from the prior to the posterior through a homotopy-based
evolution, defined by
\[
p_{z}(x;\tau)=\frac{p(x)[g(z|x)]^{\tau}}{K(\tau)},
\]
where $\tau\in[0,1]$ is a continuous homotopy variable, and $K(\tau)=\int p(x)g^{\tau}(z|x)dx$
is the normalizing constant. As $\tau$ varies from $0$ to $1$,
the prior $p_{z}(\cdot;0)$ transforms smoothly to the posterior $p_{z}(\cdot;1)$.
This transformation can be realized by modeling the Particle Flow
as a stochastic differential equation (SDE), or Itô process:
\[
dx=f_{z}(x,\tau)d\tau+q_{z}(x,\tau)dw_{\tau},
\]
that obeys the Fokker-Plank or Kolmogorov forward equation
\[
\frac{\partial p_{z}}{\partial\tau}=\triangledown\cdot(f_{z}p_{z})+\frac{1}{2}\nabla^{2}(Q_{z}p_{z}),
\]
where $f_{z}(\cdot,\tau)$ is the drift vector field in {\small$\mathbb{X}$},
$w_{\tau}$ is a standard Brownian motion or Wiener process in {\small$\mathbb{X}$}
with $\mathbb{E}[dw_{\tau}dw_{\tau}^{T}]=I_{n}d\tau$, $q_{z}(x,\tau)$
is a diffusion matrix, and $Q_{z}(x,\tau)=q_{z}(x,\tau)q_{z}^{T}(x,\tau)$
is the diffusion covariance. 

The Fokker-Plank SDE formulation provides flexibility for incorporating
different design criteria into the flow, with the central objective
of determining a suitable drift $f_{z}$ and diffusion covariance
$Q_{z}$ (which is a symmetric positive semidefinite matrix). Solutions
to the Particle Flow problem---though generally non-unique---ensure
the state $x$ at (homotopy) time $\tau$ has density $p_{z}(x,\tau)$.
Suppose that $x(0)$ is a sample from the prior $p(\cdot)=p_{z}(\cdot;0)$,
and is the initial condition for the flow, then $x(1)=\int_{0}^{1}f_{z}(x,\tau)d\tau+\int_{0}^{1}q_{z}(x,\tau)dw_{\tau}$
is a sample from the posterior $p_{z}(\cdot;1)$. Integration can
be carried out numerically, for example with an Euler-Maruyama scheme,
or any other suitable variant for SDEs.

Among the many Particle Flow implementations, a notable development
is the Gromov flow \cite{DaumHuangNoushin2018Gromov}, later generalized
into a more efficient scheme \cite{DaumDaiHuangNoushin2024Bulletproofing},
using
\begin{eqnarray*}
f_{z} & = & \left[\triangledown^{2}l_{p,z}\right]^{-1}\left[-\triangledown^{T}l_{g,z}+K_{z}\triangledown^{T}l_{p,z}\right],\\
K_{z} & = & \frac{1}{2}\left[\triangledown^{2}l_{p,z}\right]Q_{z}+\frac{1}{2}\left[\triangledown^{2}l_{g,z}\right]\left[\triangledown^{2}l_{p,z}\right]^{-1},\\
Q_{z} & = & -2\sqrt{\left|tr\left(\triangledown f_{z}^{*}\right)\right|/\mu}\left[\triangledown^{2}l_{p,z}\right]^{-1},
\end{eqnarray*}
where $l_{p,z}=log\left(p_{z}\right)$, $l_{p}=log\left(p\right)$,
$l_{g,z}=log\left(g(z|\cdot)\right)$, $\triangledown^{2}l_{p,z}=\triangledown^{2}l_{p}+\lambda\triangledown^{2}l_{g,z}$
, $f_{z}^{*}$ is the flow $f_{z}$ with $K_{z}=\boldsymbol{0}$,
and $\mu>0$ is a user parameter that controls the truncation error.
Notice that the construction of the Particle Flow is independent of
the normalizing constant $K(\tau)$. 

\bibliographystyle{IEEEtran}
\bibliography{reflib}

\end{document}